\documentclass[12pt, a4paper]{article}
\pdfoutput=1
\usepackage{jheppub}
\usepackage{subcaption}
\usepackage[T1]{fontenc}
\usepackage{float}
\usepackage{latexsym}
\usepackage{amsmath,amstext,amssymb,amsfonts,
amscd,bm,array,multirow,amsbsy,mathrsfs,amsthm,calc}
\usepackage{t1enc}
\usepackage{indentfirst}
\usepackage{pb-diagram}
\usepackage{graphicx}
\usepackage{enumerate}
\usepackage{color}
\usepackage[all]{xy}
\usepackage{hyperref}
\hypersetup{colorlinks=false,pdfborderstyle={/S/U/W 0}}
\usepackage[usenames,dvipsnames]{xcolor}
\usepackage{url}
\usepackage{upgreek}

\theoremstyle{plain}
\newtheorem{thm}{Theorem}[section]             

\newtheorem{prop}[thm]{Proposition}

\newtheorem{cor}[thm]{Corollary}

\theoremstyle{definition}
\newtheorem{definition}[thm]{Definition}

\theoremstyle{remark}
\newtheorem{remark}[thm]{Remark}
\newtheorem{eg}[thm]{Example}

\newcommand{\be}{\begin{equation*}}
\newcommand{\ee}{\end{equation*}}
\newcommand{\ben}{\begin{equation}}
\newcommand{\een}{\end{equation}}
\newcommand{\beqa}{\begin{eqnarray*}}
\newcommand{\eeqa}{\end{eqnarray*}}
\newcommand{\beqan}{\begin{eqnarray}}
\newcommand{\eeqan}{\end{eqnarray}}
\newcommand{\nn}{\nonumber}

\def\i{\mathbf{i}}

\def\oR{{\overline \R}}
\def\N{\mathbb{N}}
\def\Z{\mathbb{Z}}

\def\R{\mathbb{R}}

\def\End{\mathrm{End}}
\def\Hess{\mathrm{Hess}}

\def\Lim{\mathrm{Lim}}
\def\Crit{\mathrm{Crit}}

\def\Ends{\mathrm{Ends}}
\def\Int{\mathrm{Int}}

\def\fM{\mathfrak{M}}

\def\ren{{\mathrm{ren}}}
\def\Par{\mathrm{Par}}

\def\rT{\mathrm{T}}
\def\rK{\mathrm{K}}

\newcommand{\pd}{\partial}
\def\dd{\mathrm{d}}

\def\bvarepsilon{\boldsymbol{\varepsilon}}

\newcommand{\id}{\mathrm{id}}
\newcommand{\Tr}{\mathrm{Tr}}

\newcommand{\sign}{\mathrm{sign}}

\def\cC{\mathcal{C}}
\def\cD{\mathcal{D}}

\def\cG{\mathcal{G}}
\def\cH{\mathcal{H}}
\def\cL{\mathcal{L}}
\def\cM{\mathcal{M}}

\def\cS{\mathcal{S}}
\def\cX{\mathcal{X}}

\def\IR{\mathrm{IR}}
\def\UV{\mathrm{UV}}

\def\Met{\mathrm{Met}}

\def\hcM{\hat{\cM}}

\def\rC{\mathrm{C}}
\def\rD{\mathrm{D}}
\def\rA{\mathrm{A}}

\def\rS{\mathrm{S}}

\def\veta{\boldsymbol{\eta}}
\def\Stab{\mathrm{Stab}}
\def\param{\mathrm{par}}
\def\fM{\mathfrak{M}}

\def\graph{\mathrm{graph}}

\newcommand{\eqdef}{\stackrel{{\rm def.}}{=}}

\def\grad{\mathrm{grad}}
\def\Sym{\mathrm{Sym}}
\def\End{\mathrm{End}}

\def\im{\mathrm{im}}
\def\O{\mathrm{O}}
\def\o{\mathrm{o}}

\def\alim{\mathrm{lim}_\alpha}
\def\olim{\mathrm{lim}_\omega}

\def\hcM{\hat{\cM}}
\def\bvert{\big{\vert}}
\def\Bvert{\Big{\Vert}}
\def\Pot{\mathrm{Pot}}

\def\ind{\mathrm{ind}}
\def\rM{\mathrm{M}}

\def\RG{\mathrm{RG}}

\def\vol{\mathrm{vol}}
\def\P{\mathbb{P}}

\def\mS{\mathbb{S}}
\def \Grad{\mathrm{Grad}}
\def\cT{\mathcal{T}}

\title{Dynamical renormalization and universality in classical multifield cosmological models}

\author{Calin Iuliu Lazaroiu}

\affiliation{Horia Hulubei National
  Institute of Physics and Nuclear Engineering,\\
  Reactorului 30, Bucharest-Magurele, 077125, Romania\\ 
  }

\emailAdd{lcalin@theory.nipne.ro}

\abstract{We study the scaling behavior of classical multifield
cosmological models with complete scalar manifold $(\cM,\cG)$ and
positive smooth scalar potential $\Phi$, introducing a dynamical
renormalization group action which relates their UV and IR limits. We
show that the RG flow of such models interpolates between a
modification of the geodesic flow of $(\cM,\cG)$ (obtained in the UV
limit) and the gradient flow of $(\cM,\cG,V)$ (obtained in the IR
limit), where the {\em classical effective potential} $V$ is
proportional to $\sqrt{2\Phi}$. Using this fact, we show that
two-field models whose scalar manifold has constant Gaussian curvature
equal to $-1$, $0$ or $1$ are {\em infrared universal} in the sense
that they suffice to describe the first order IR approximants of 
cosmological orbits for all two-field models with positive smooth
scalar potential.}

\begin{document}

\maketitle

\pagebreak

\section*{Introduction}

The study of cosmological models with more than one real scalar field
(known as {\em multifield cosmological models}) is crucial for
connecting fundamental theories of gravity and matter to the early
universe, since the effective description of the generic string or
M-theory compactification contains many such fields. In particular,
multifield models are important for cosmological applications of the
swampland program \cite{V, OV} (see \cite{BCV, BCMV} for reviews),
which drew attention to this often overlooked area of cosmology
\cite{AP, OOSV,GK}. 

Despite their importance, the current theory of such models is poorly
developed, especially as pertains to deeper conceptual and
mathematical aspects. In fact, most contributions to the area focus on
analyzing the leading orders in the formal cosmological perturbation
expansion of such models around a ``God given'' abstract cosmological
curve \cite{SS,NT,GT,GLM,EST,Pinol,Pinol2field} (often considering
only two-field models and under highly restrictive conditions such as
assuming validity of the SRST approximation of \cite{PT1,PT2}) or on
investigations of specific models and cosmological curves aimed at
illustrating the feasibility of various inflation scenarios. Since the
class of multifield models is continuously infinite, one is left
wondering if a more systematic approach is possible which
could permit fundamental progress in the field.

At the classical level (and before considering fluctuations)
multifield models are described by ODEs to which the powerful methods
of the geometric theory of dynamical systems \cite{Palis,Katok,Shub}
could be applied. Even at this level, the current literature
remains largely concerned with local descriptions and pays little
attention to global aspects, thus hardly amounting to a
geometric theory. In particular, there were only a few attempts to
develop an approach to the subject that could place it on a firm
mathematical foundation and hence could serve as a stepping stone
toward the program of connecting string theory to the early universe
in a systematic manner. The purpose of this paper is to address some
of these limitations by proposing a conceptual viewpoint on the
classical dynamics of such models.

As pointed out in \cite{genalpha}, classical multifield cosmological
models admit a precise global description which leads to their
formulation as geometric dynamical systems. Mathematically, such a
model is parameterized by the reduced Planck mass $M$ (equivalently,
by the {\em rescaled} Planck mass $M_0\eqdef M\sqrt{\frac{2}{3}}$) and
by a {\em scalar triple} $(\cM,\cG,\Phi)$, where $\cM$ is a (generally
non-compact) connected manifold, $\cG$ is a Riemannian
metric on $\cM$ and $\Phi$ is a real-valued function defined on $\cM$;
we assume that these three objects are smooth. Such data enters the
classical action of the scalar fields, which is described by a
nonlinear sigma model defined on spacetime, where:
\begin{itemize}
\item The {\em target space} $\cM$ is the space in which the scalar
fields take values, namely they are described as smooth maps
from spacetime into this manifold.
\item The {\em scalar field metric} $\cG$ governs the kinetic energy
of the scalar fields.
\item The {\em scalar potential} $\Phi$ governs the potential energy
of the scalar fields.
\end{itemize}
To ensure conservation of energy, one requires the metric $\cG$ to be
complete on $\cM$; this precludes existence of scalar field
configurations which ``disappear into nothing''. To preclude global
instability of dynamics, one requires that $\Phi$ is bounded from
below by zero; for simplicity, we will assume throughout this paper that
$\Phi$ is {\em strictly} positive. The complete Riemannian manifold
$(\cM,\cG)$ is called the {\em scalar manifold}.

The classical cosmological model is obtained from this data by
considering gravity at rescaled Planck mass $M_0$ coupled to the sigma
model parameterized by $(\cM,\cG,\Phi)$. Before considering
perturbations, one takes the spacetime to be of FLRW form with
cosmological time $t$ and conformal factor $a$, while restricting the
scalar fields to depend only on $t$ and hence to be described by a
curve $\varphi:I\rightarrow \cM$, where $I$ is a non-degenerate
interval (i.e. an interval on the real axis which is not empty or
reduced to a point). With these restrictions, the Einstein equations
and scalar field equations of motion reduce to a single second order
ODE (known as the {\em cosmological equation}) for $\varphi$, whose
solutions we call {\em cosmological curves}. Since it is geometric in
the sense of \cite{Chern}, the cosmological equation is equivalent to
a dynamical system defined on the total space of the tangent bundle
$T\cM$ of the scalar manifold $\cM$ by a certain vector field $S$
which we call the {\em cosmological semispray}. The latter lives on
$T\cM$ and is parameterized by $M_0$, by the scalar field metric $\cG$
and by the scalar potential $\Phi$. The equivalence between the
cosmological equation and the cosmological dynamical system follows
from the theory of geometric second order ODEs \cite{Kosambi, Cartan,
Chern, SLK, Bucataru} (see \cite{Nis} for a brief account). The flow
defined by this dynamical system on $T\cM$ is called the {\em
cosmological flow} of the scalar triple $(\cM,\cG,\Phi)$ at rescaled
Planck mass $M_0$.

Since the topology of $\cM$ is arbitrary, a systematic approach to
such models requires the full force of the geometric theory of
dynamical systems as developed for example in \cite{Palis,Katok,Shub};
one cannot give a satisfactory treatment merely by studying the
description of this system in local coordinates on the scalar manifold
$\cM$. For example, $\cM$ need not be simply connected so a maximal
cosmological curve need not be contractible. More importantly, $\cM$
is generally non-compact in physically interesting applications, so a
cosmological curve can ``escape to infinity'' in the sense that it can
approach a Freudenthal end \cite{Freudenthal, Morita, Isbell, Peschke}
of $\cM$ for early or late cosmological times depending on the
behavior of $\Phi$ and $\cG$ near that end. Non-compactness of $\cM$
complicates the classification of past and future limit points of
cosmological curves, which governs the early and late time behavior of
the model -- an aspect which is of direct physical interest. It also
prevents the cosmological flow from being future complete unless
appropriate conditions are imposed on $\Phi$ and $\cG$ in the vicinity
of the Freudenthal ends of $\cM$. Such global aspects of cosmological
dynamics have direct implications for any attempt to construct
effective descriptions when one incorporates quantum effects (as
envisaged for example in cosmological applications of the swampland
program). Indeed, any effective description depends on having an a
priori topological classification of maximal cosmological curves since one
cannot expect an effective description to be the same in every
topological class. Intuitively, the topological classification of
maximal cosmological curves partitions the dynamics of the model into
``phases'' and every phase will have its own effective description.

The global aspects mentioned above are already important for two-field
cosmological models i.e. when $\cM$ is a surface, in which case we
denote it by $\Sigma$. The topological classification of borderless
connected paracompact surfaces \cite{Kerekjarto, Stoilow, Richards}
through their orientability, reduced genus and space of ends (together
with its length two chain of distinguished subspaces) shows that the
global theory of cosmological dynamical systems is already very rich
in the two-field case. This is especially true when the set of ends is
infinite, in which case it can be for example a Cantor space.  The
theory is quite rich even when one restricts to oriented surfaces
which are topologically finite in the sense that they have a
finitely-generated fundamental group and hence a finite number of
Freudenthal ends. The dynamical complexity of two-field models having
such surfaces as targets was illustrated in our previous work
\cite{genalpha,elem, modular} (see \cite{unif} for a brief
review) when $\cG$ has constant negative curvature; this corresponds
to {\em two-field generalized $\alpha$-attractor models}, which form a
very wide extension of the topologically trivial class of Poincar\'e
disk two-field models considered \cite{KLR}.

Since the connected manifold $\cM$, the complete Riemannian metric
$\cG$ and the positive potential $\Phi$ are arbitrary, the task of
classifying multifield cosmological dynamics may seem hopeless at
first sight. From a dynamical systems perspective, one could attempt
to classify cosmological flows up to the appropriate notions of conjugation
or equivalence, a task which is nontrivial when $\cM$ is non-compact
and was not yet attempted in the generality considered here.

In this paper, we propose a different approach to extract the
essential features of such models and to organize them into
qualitatively distinct dynamical classes. Our point of view is
inspired by ideas akin to those used in the theory of critical
phenomena, with the points of $T\cM$ playing the role of microscopic
states and $\cG$ and $\Phi$ playing a role similar to that of
non-equilibrium thermodynamic observables. A state of the system is a
point $u(t)$ of $T\cM$ which evolves in time according to the
cosmological flow and has ``thermodynamic'' parameters
$\frac{1}{2}||u(t)||_\cG^2=\frac{1}{2}\cG(\pi(u(t)))(u(t),u(t))$ and
$\Phi(\pi(u(t)))$, where $\pi$ is the projection of $T\cM$. As in
thermodynamics, knowing the values of these parameters does not
determine the ``microscopic'' state $u(t)$. Since the cosmological
equation is autonomous, the system is stationary in the sense that its
flow is invariant under shifts of the cosmological time by an
arbitrary constant. Following the analogy with critical phenomena, we
consider the behavior of the model under scale transformations
$t\rightarrow t/\epsilon$ of the cosmological time (where $\epsilon$
is a positive parameter), showing that such transformations induce a
renormalization group action on $M_0$, $\cG$ and $\Phi$. The {\em
scaling limits} $\epsilon\rightarrow \infty$ and $\epsilon\rightarrow
0$ capture the high frequency (or ultraviolet) and low frequency (or
infrared) behavior of cosmological curves in the sense that they
``isolate'' the high and low frequency characteristic oscillations of
such curves. We then show that taking $\epsilon$ to be large or small
corresponds respectively to replacing the cosmological flow of
$(\cM,\cG,\Phi)$ with a modification of the geodesic flow of the
scalar manifold $(\cM,\cG)$ or with the gradient flow of the {\em
classical effective potential} $V=M_0\sqrt{2\Phi}$ on this Riemannian
manifold. The ultraviolet limit is scale invariant, while the infrared
limit (which we consider up to first order in the scale parameter
$\epsilon$) is invariant under Weyl transformations of $\cG$ up to
reparameterization of the flow curves; in particular the first order
IR approximation of cosmological orbits is Weyl-invariant. This limit
is degenerate in that it confines the cosmological flow to the graph
of the vector field $-\grad_\cG V$ inside $T\cM$; accordingly, the order of the
cosmological equation drops by one in the infrared limit. The scaling
limits arise as the leading orders of systematic asymptotic
approximations (called the UV and IR expansions) of the cosmological
flow around the geodesic flow of $(\cM,\cG)$ or around the gradient
flow of $(\cM,\cG,V)$. Such expansions are natural from a physics
perspective. They are also mathematically natural since geodesic and
gradient flows are well-studied subjects in the theory of dynamical
systems -- though the generic non-compactness of $\cM$ complicates the
analysis. Use of the scaling limits allows us to classify multifield
cosmological models into UV and IR {\em universality classes} whose
study relates to that of such classical flows.

The Weyl-invariance of infrared approximants to cosmological orbits
has striking implications for two-field models. In this case, the
uniformization theorem of Poincar\'e \cite{UnifT} states that the Weyl
equivalence class of $\cG$ contains a unique complete metric $G$
(called the {\em uniformizing metric}) which has constant Gaussian
curvature $K$ equal to $-1$, $0$ or $+1$. Thus the first order
infrared approximants of cosmological orbits for any cosmological
two-field model with positive scalar potential $\Phi$ coincide with
those of the model obtained by replacing $\cG$ with $G$ and with the
gradient flow orbits of the scalar triple $(\cM,G,V)$. In particular,
models whose scalar manifold metric has constant Gaussian curvature
provide distinguished representatives of the infrared universality
classes of all two-field models.

The generic case $K=-1$ arises whenever $\Sigma$ is a (compact or
non-compact) surface of general type, i.e.  not diffeomorphic with a
2-plane, a 2-sphere, a real projective plane, a 2-torus, a Klein
bottle, an open 2-cylinder or an open M\"{o}bius strip.  In this
situation, the uniformizing metric $G$ is hyperbolic. When $\Sigma$ is
diffeomorphic with $\rS^2$ or $\R\P^2$, the metric $G$ has Gaussian
curvature $+1$, while when $\Sigma$ is diffeomorphic with a torus or
Klein bottle the uniformizing metric is flat and complete. When
$\Sigma$ is {\em exceptional}, i.e. diffeomorphic with a plane, an
open annulus or an open M\"{o}bius strip, the metric $\cG$ uniformizes
to a complete flat metric or to a hyperbolic metric depending on its
conformal class. In this situation, a description of universality
classes requires considering {\em both}\footnote{For the three
exceptional surfaces, a hyperbolic metric on $\Sigma$ is conformally
flat but it is not conformally equivalent with a {\em complete} flat
metric.}  models with complete flat and hyperbolic scalar manifold
metric. Notice that two-field models with contractible target are of
exceptional type and hence the Poincar\'e disk models of \cite{KLR}
cannot be IR universal among such. The qualitatively different
behavior of two-field models with distinct target space topology
illustrates the importance of global aspects in cosmological dynamics.

The paper is organized as follows. In Section \ref{sec:models}, we
recall the global description of multifield cosmological models,
outline its dynamical system formulation and some of its basic
properties and discuss a universal two-parameter group of similarities
of the cosmological equation. We also describe two natural equivalence
relations on multifield cosmological models which arise from
underlying groupoid structures and make some observations on the early
and late time behavior of cosmological curves for models with
non-compact target manifold. In Section \ref{sec:ScalingLimits}, we
discuss the scale transformations and scaling limits of multifield
cosmological models, showing that the UV and IR limits recover
respectively a modification of the geodesic flow of $(\cM,\cG)$ and
the gradient flow of $(\cM,\cG,V)$, where $V=M_0\sqrt{2\Phi}$. We also
derive consistency conditions for the UV and IR approximations,
showing that they differ from other approximations commonly used in
cosmology (such as the slow roll and SRST approximations or the
gradient flow approximation of \cite{genalpha}). In Section
\ref{sec:ren}, we introduce the dynamical renormalization group of
such models. We show invariance of first order IR approximant orbits
under Weyl transformations of the scalar manifold metric, define IR
universality classes and briefly discuss the late time infrared phase
structure.  In Section \ref{sec:2field} we prove IR universality of
two-field models whose scalar manifold metric has constant Gaussian
curvature equal to $-1$, $0$ or $1$. Section \ref{sec:Conclusions}
presents our conclusions and some directions for further research. The
Appendices contain some technical notions and results used in the main
text.

\paragraph{Notations and conventions.}

All target manifolds $\cM$ considered in this paper are connected,
smooth, Hausdorff and paracompact. If $V$ is a smooth real-valued
function defined on $\cM$, we denote by: \be \Crit V\eqdef
\{c\in \cM| (\dd V)(c)=0\} \ee the set of critical points of $V$. For
any $c\in \Crit V$, we denote by $\Hess(V)(c)\in \Sym^2(T^\ast_c\cM)$
the Hessian of $V$ at $c$, which is a well-defined and coordinate
independent symmetric bilinear form defined on the tangent space
$T_c\cM$. Recall that a critical point
$c$ of $V$ is called {\em nondegenerate} if $\Hess(V)(c)$ is a
non-degenerate bilinear form. When $V$ is a Morse function (i.e. all
of its critical points are non-degenerate), the set $\Crit V$ is
discrete.

We denote by $\hcM$ the Freudenthal (a.k.a. end) compactification of
$\cM$, which is a Hausdorff topological space containing $\cM$ as a
dense subset (see \cite{Freudenthal,Morita,Isbell}). A metric $\cG$ on
$\cM$ is called {\em hyperbolic} if it is complete and of constant
sectional curvature equal to $-1$. In particular, a metric defined on
a surface is hyperbolic if it is complete and of Gaussian curvature
$-1$.

\section{Multifield cosmological models}
\label{sec:models}

Throughout this paper, a {\em multifield cosmological model} means a
classical cosmological model with a finite number $d>1$ of scalar
fields, which is derived from the following matter-gravity action on a
spacetime with topology $\R^4$:
\ben
\label{S}
S[g,\varphi]=\int \vol_g \cL[g,\varphi]~~,
\een
where:
\ben
\label{cL}
\cL[g,\varphi]=\frac{M^2}{2} \mathrm{R}(g)-\frac{1}{2}\Tr_g \varphi^\ast(\cG)-\Phi\circ \varphi~~.
\een
Here $M$ is the reduced Planck mass, $g$ is the spacetime metric on
$\R^4$ (taken to be of ``mostly plus'') signature, while $\vol_g$ and
$\mathrm{R}(g)$ are the volume form and Ricci scalar of $g$. The scalar fields
are described by a smooth map $\varphi:\R^4\rightarrow \cM$, where
$\cM$ is a (generally non-compact) connected, smooth and paracompact
manifold of dimension $d$ which is endowed with a smooth Riemannian
metric $\cG$, while $\Phi:\cM\rightarrow \R$ is a smooth function
which plays the role of potential for the scalar fields. As mentioned
in the introduction, we require that $\cG$ is complete to ensure
conservation of energy. For simplicity, we will also assume that
$\Phi$ is strictly positive on $\cM$. The quantity
$\Tr_g\varphi^\ast(\cG)$ is the trace of the $(1,1)$-tensor obtained
by raising one of the indices of the covariant 2-tensor
$\varphi^\ast(\cG)$ with respect to the spacetime metric $g$, while
$\Phi\circ \varphi:\R^4\rightarrow \R$ is the standard mathematical
notation\footnote{This is commonly written as $\Phi(\varphi)$ in the
physics literature, though this notation is misleading since $\varphi$
is not the argument of $\Phi$.} for the real-valued function defined
on $\R^4$ which is obtained by composing $\Phi$ with $\varphi$:
\be
(\Phi\circ\varphi)(x^0,\ldots, x^3)=\Phi(\varphi(x^0,\ldots, x^3))~~.
\ee
In local coordinates on $\cM$, we have $\varphi(x^0,\ldots,
x^3)\!=\!(\varphi^1(x^0,\ldots, x^3),\ldots, \varphi^d(x^0,\ldots, x^3))$.
The second term in the Lagrangian above takes the familiar
sigma model form if one uses local coordinates on $\cM$:
\be
\frac{1}{2}(\Tr_g \varphi^\ast(\cG))(x)=\frac{1}{2}g^{\mu \nu}(x) \cG_{\alpha\beta}(\varphi(x)) \partial_\mu \varphi^\alpha\partial_\nu \varphi^\beta~~.
\ee
Notice that the action \eqref{S} and its Lagrangian density \eqref{cL}
are manifestly geometric, i.e. written in coordinate-free form and
without making any restrictive assumptions on the differential
topology of $\cM$, which (except for being connected) can be arbitrary
since any paracompact manifold admits Riemannian metrics. Also notice
that such a model is parameterized by the quadruplet $\fM\eqdef
(M_0,\cM,\cG,\Phi)$.

\subsection{The cosmological equation and cosmological dynamical system} 
\label{subsec:eom}

\noindent The multifield cosmological model parameterized by
$(M_0,\cM,\cG,\Phi)$ is obtained by assuming that $g$ is an FLRW
metric with flat spatial section:
\ben
\label{FLRW}
\dd s^2_g=-\dd t^2+a(t)^2\sum_{i=1}^3 \dd x_i^2
\een
(where $a(t)>0$) and that $\varphi$ depends only on
$t\eqdef x^0$, which we call {\em cosmological time}. In
this case, the equations of motion derived from \eqref{S} (namely the
Einstein equations and the equation of motion for $\varphi$) amount to
the following system of coupled nonlinear ODEs:
\beqan
\label{EL}
\nabla_t \dot{\varphi}+3H \dot{\varphi}+(\grad_{\cG} \Phi)\circ \varphi &=&0 \nn\\
\frac{1}{3}\dot{H}+H^2 - \frac{\Phi\circ \varphi}{3M^2} &=& 0\\
\dot{H} +\frac{||\dot{\varphi}||_\cG^2}{2M^2} &=& 0~~,\nn
\eeqan
where the dot indicates derivation with respect to $t$ and $H\eqdef
\frac{\dot{a}}{a}\in \cC^\infty(\R)$ is the {\em Hubble parameter}.
The last relation in the system above is called the {\em Friedmann
equation}. Notice that $a$ enters this system only through its
logarithmic derivative $H$.

\paragraph{The cosmological equation.}

When $H$ is positive (which we assume throughout this paper), it can
be eliminated algebraically using last two equations, which give:
\ben
\label{Hvarphi}
H(t)=H_\varphi(t)\eqdef \frac{1}{3 M_0}\left[||\dot{\varphi}(t)||_\cG^2+2\Phi(\varphi(t))\right]^{1/2}~~,
\een
where we defined the {\em rescaled Planck mass} $M_0$ though:
\ben
M_0\eqdef M\sqrt{\frac{2}{3}}~~.
\een
Eliminating $H$ though \eqref{Hvarphi} allows one to reduce the system
\eqref{EL} to the following autonomous geometric second order
ODE, which we call the {\em cosmological equation}:
\ben
\label{eomsingle}
\nabla_t \dot{\varphi}(t)+\frac{1}{M_0} \left[||\dot{\varphi}(t)||_\cG^2+2\Phi(\varphi(t))\right]^{1/2}\dot{\varphi}(t)+ (\grad_{\cG} \Phi)(\varphi(t))=0~~.
\een
Here $\nabla_t\eqdef \nabla_{\dot{\varphi}(t)}$. The cosmological
equation is {\em equivalent} with the system \eqref{EL} when $H>0$.
Since $\grad_{\cG}\Phi$ is invariant under transformations of the form
$(\cG,\Phi)\rightarrow (\lambda\cG,\lambda\Phi)$ with $\lambda$ a
positive constant, this equation can be written as:
\be
\nabla_t \dot{\varphi}(t)+\left[||\dot{\varphi}(t)||_{\cG_0}^2+2\Phi_0(\varphi(t))\right]^{1/2}\dot{\varphi}(t)+ (\grad_{\cG_0} \Phi_0)(\varphi(t))=0~~,
\ee
while \eqref{Hvarphi} reads:
\be
H_\varphi(t)=\frac{1}{3}\left[||\dot{\varphi}(t)||_{\cG_0}^2+2\Phi_0(\varphi(t))\right]^{1/2}~~,
\ee
where we defined the {\em rescaled scalar field metric} and {\em
rescaled scalar potential} through:
\be
\cG_0\eqdef \frac{1}{M_0^2}\cG~~\mathrm{and}~~\Phi_0\eqdef\ \frac{1}{M_0^2}\Phi~~.
\ee
In particular, the cosmological equation
depends only on the {\em rescaled scalar triple} $(\cM,\cG_0,\Phi_0)$.
Given a solution $\varphi:I\rightarrow \cM$ of this equation, relation
\eqref{Hvarphi} determines the scale factor $a(t)$ up to a
multiplicative constant $C>0$:
\be
a(t)=C e^{\int_{t_0}^t \dd t' H_\varphi(t')}~~,
\ee
where $t_0\in I$ is chosen arbitrarily.

\paragraph{Cosmological curves and cosmological orbits.}

The solutions $\varphi:I\rightarrow \cM$ of \eqref{eomsingle} (where
$I$ is a non-degenerate interval) are called {\em cosmological
curves}. The images $\varphi(I)$ of these curves in $\cM$ will be
called {\em cosmological orbits}. Notice that a cosmological orbit
need not be an immersed submanifold of $\cM$. It can be shown that the
singular points of a cosmological curve (i.e. those points
where its tangent vector vanishes) form an at most countable set whose
complement in the corresponding orbit is a union of mutually disjoint
embedded submanifolds of $\cM$. The cosmological times corresponding
to the singular points form a discrete subset of its
interval of definition. 

\paragraph{The cosmological semispray.}

Since the cosmological equation is manifestly geometric, it is also
geometric in the weaker sense that its local description in a
coordinate system on $\cM$ is invariant under changes of local
coordinates. It follows that this equation is equivalent with the flow
equation of a special type of vector field (called a {\em semispray}
or {\em second order tangent vector field}) defined on the total space
of the tangent bundle of $\cM$. We refer the reader to \cite{SLK,
Bucataru} for an introduction to the well-developed theory of
semisprays and second order geometric ODEs; for completeness, let us
recall the relevant definitions and properties of such vector fields.

Let $J\in \End_{T\cM}(TT\cM)$ be the canonical endomorphism
(a.k.a. tangent structure) of the double tangent bundle $TT\cM$ and
$C\in \cX(T\cM)$ be its Euler-Liouville vector field. Let
$\kappa:TT\cM\rightarrow TT\cM$ be the canonical involution
(a.k.a. flip) of $TT\cM$, $\pi:T\cM\rightarrow \cM$ be the bundle
projection of $T\cM$ and $\pi_T:TT\cM\rightarrow T\cM$ be the tangent
bundle projection (a.k.a. first projection) of $TT\cM$. Then the map
$\pi_\ast:TT\cM\rightarrow T\cM$ endows $TT\cM$ with a second vector
bundle structure such that $(TT\cM,\pi_T,\pi_\ast,\cM)$ is a double
vector bundle. By definition, a second order tangent vector to $\cM$
is an element $w\in TT\cM$ such that $J(w)=C_w$. This condition is
equivalent with $\kappa(w)=w$ and also with $\pi_\ast(w)=\pi_T(w)$. We
denote by $T^sT\cM$ the vector sub-bundle of $TT\cM$ consisting of
second order tangent vectors. By definition, a smooth semispray on
$T\cM$ is a smooth vector field $S\in \cX(T\cM)$ which is a section of
this sub-bundle, i.e. which has the property that $S_u$ is a double
tangent vector to $\cM$ for all $u\in T\cM$. This amounts to requiring
that $S$ satisfies the equivalent conditions:
\be
J(S)=C\Longleftrightarrow \kappa(S)=S\Longleftrightarrow \pi_\ast(S)=\pi_T\circ S~~,
\ee
where in the right hand side of the last equality $S$ is viewed as a
map from $T\cM$ to $TT\cM$.

The translation between the cosmological equation and the integral
curve equation of the cosmological semispray is performed by
considering the {\em canonical lift} of a cosmological curve
$\varphi:I\rightarrow \cM$, which is defined as its first jet
prolongation:
\be
c(\varphi)\eqdef j^1(\varphi)=\dot{\varphi}:I\rightarrow T\cM~~,
\ee
where $\dot{\varphi}(t)\in T_{\varphi(t)}\cM$ is the tangent vector to
$\varphi$ at cosmological time $t$. The canonical lift gives
an injective map:
\be
c:C(\cM)\rightarrow C(T\cM)
\ee
from the set $C(\cM)$ of smooth curves in $\cM$ to the set $C(T\cM)$
of smooth curves in $T\cM$. This is a right inverse of the map
$p:C(T\cM)\rightarrow C(\cM)$ defined through:
\be
p(\gamma)\eqdef \pi\circ \gamma:I\rightarrow \cM
\ee
for any smooth curve $\gamma:I\rightarrow T\cM$. The relation $p\circ
c=\id_{C(\cM)}$ reflects the fact that $\dot{\varphi}(t)\in T\cM$ is
a tangent vector to $\cM$ at the point $\varphi(t)$. Since
the canonical lift $\varphi$ satisfies:
\be
\pi_\ast\circ c(c(\varphi))=\pi_T\circ c(\varphi)~~,
\ee
the image of the map $c$ coincides with the subset $C_s(T\cM)$ of
$C(T\cM)$ consisting of {\em semispray curves} in $T\cM$,
which are defined as those smooth curves $\gamma:I\rightarrow T\cM$
whose tangent vector at any point is a double tangent vector to $\cM$,
i.e. which satisfy $\dot{\gamma}(t)\in T^s_{\gamma(t)}T\cM$ for all
$t\in I$. The inverse of the bijection $c:C(\cM)\rightarrow C_s(T\cM)$
is the restriction of $p$ to $C_s(T\cM)$. 

The canonical lifts of cosmological curves are called cosmological
{\em flow} curves. The {\em cosmological semispray} $S\in \cX(T\cM)$
is defined by the property that its integral curves coincide with the
cosmological flow curves. It can be shown that $S$ is given by:
\be
S=\mS-Q~~,
\ee
where $\mS$ is the geodesic spray of the scalar manifold
$(\cM,\cG)$ and the {\em cosmological correction} $Q\in \cX(T\cM)$
is the vertical vector field:
\ben
\label{Q}
Q=\cH C+(\grad_{\cG} \Phi)^v~~.
\een
Here the superscript $v$ denotes the vertical lift of vector fields
from $\cM$ to $T\cM$ and $\cH:T\cM\rightarrow \R_{>0}$ is the {\em
reduced Hubble function} of $(\cM,\cG,\Phi)$, which is defined
through:
\ben
\label{cH}
\cH(u)\eqdef \frac{1}{M_0}\sqrt{||u||_{\cG}^2+2\Phi(\pi(u))}~~\forall u\in T\cM~~.
\een

\paragraph{The cosmological dynamical system and cosmological flow.}

The cosmological semispray $S$ defines an autonomous geometric dynamical
system $(T\cM,S)$ on the total space of the tangent bundle to $\cM$,
whose flow $\Pi:\cD\rightarrow T\cM$ is called the {\em cosmological
flow} of the model parameterized by $(M_0,\cM,\cG,\Phi)$
(equivalently, of the rescaled scalar triple
$(\cM,\cG_0,\Phi_0)$). Here $\cD\subset \R\times T\cM$ is the maximal
domain of definition of the flow. Little is known about the global
behavior of the cosmological flow of a general non-compact scalar
triple and especially about its early and late time behavior.

\paragraph{Classical cosmological observables.}

A {\em basic local cosmological observable} is a smooth function
$F:T\cM\rightarrow \R$. The {\em on-shell reduction} of $F$ on a
cosmological flow curve $\gamma:I\rightarrow T\cM$ is the function:
\be
{\hat F}_\gamma\eqdef F\circ \gamma:I\rightarrow \R~~,
\ee
while its on-shell reduction on a cosmological curve
$\varphi:I\rightarrow \cM$ is the function:
\be
F_\varphi\eqdef {\hat F}_{\dot{\varphi}}=F\circ \dot{\varphi}=F\circ c(\varphi):I\rightarrow \R~~.
\ee
One can also consider smooth functions $F:J^k(\cM)\rightarrow \R$,
where $k\geq 2$ and $J^k(\cM)$ is the $k$-th jet bundle of $\cM$. The
reduction of such a function on a cosmological curve
$\varphi:I\rightarrow \cM$ is defined though:
\be
F_\varphi\eqdef F\circ j^k(\varphi):I\rightarrow \R~~,
\ee
where $j^k(\varphi):I\rightarrow J^k(\cM)$ is the $k$-th jet
prolongation of $\varphi$.  Repeated use of the cosmological equation
allows one to express $F_\varphi$ as:
\be
F_\varphi={\tilde F}_\varphi~~,
\ee
where ${\tilde F}\in \cC^\infty(T\cM)$ is a basic local observable
constructed from $F$ and the jets $j^{k-1}(\cG)\in J^{k-1}(\Sym^2(T\cM),\R)$
and $j^{k-1}(\Phi)\in J^{k-1}(\cM,\R)$ of $\cG$ and $\Phi$.

\paragraph{Dissipativity and stationary points.}

It is easy to check that the stationary points of the cosmological
flow are the images of the critical points of $\Phi$ through the zero
section of $T\cM$. Accordingly, the stationary set of the cosmological
model coincides with the {\em trivial lift}:
\be
(\Crit \Phi)_0\eqdef \{0_c\vert c\in \Crit\Phi\}\subset T\cM
\ee
of the critical set:
\be
\Crit\Phi\eqdef \{c\in \cM\vert (\dd\Phi)(c)=0\}
\ee
of $\Phi$, whose complement:
\be
\cM_0\eqdef \cM\setminus \Crit\Phi\subset \cM
\ee
is an open submanifold of $\cM$ called the {\em noncritical set} of
the model.

A cosmological flow curve is constant iff its orbit meets a point of
the set $(\Crit \Phi)_0$, in which case its orbit is reduced to that
point. Accordingly, a cosmological curve is constant iff its orbit
meets a critical point of $\Phi$ {\em at zero speed}, in which case
the orbit coincides with that critical point. When $\Phi$ is a Morse
function, a straightforward computation shows that the stationary
points of the cosmological flow are hyperbolic.

It is also easy to see that the cosmological flow is
dissipative. For this, consider the {\em total energy} observable
$E:T\cM\rightarrow \R$ defined through:
\be
E(u)\eqdef \frac{1}{2}||u||_\cG^2+\Phi(\pi(u))~~\forall u\in T\cM~~,
\ee
which is the sum of the kinetic and potential energies:
\be
E^k(u)\eqdef \frac{1}{2}||u||_\cG^2~~,~~E^p(u)\eqdef \Phi(\pi(u))
\ee
and is related to the reduced Hubble function by
$\cH=\frac{1}{M_0}\sqrt{2E}$. The last equation in \eqref{EL} implies
that the total energy of a cosmological curve $\varphi$:
\ben
\label{energy}
E_\varphi(t)\eqdef E(\dot{\varphi}(t))=\frac{1}{2}||\dot{\varphi}(t)||_\cG^2+\Phi(\varphi(t))=\frac{9 M_0^2}{2} H_\varphi(t)^2
\een
satisfies:
\ben
\label{EnergyEq}
\frac{\dd E_\varphi(t)}{\dd t}=-\frac{1}{M_0}\sqrt{2E_\varphi(t)}||\dot{\varphi}(t)||_\cG^2
\een
and hence decreases strictly with time when $\varphi$ is non-constant.
In particular, any non-constant cosmological curve is aperiodic and
without self-intersections. Moreover, all non-constant cosmological
flow curves are embedded curves in $T\cM$. As mentioned before,
cosmological curves need not be immersed in $\cM$ but their singular
times form an at most countable discrete subset of their interval of
definition. A cosmological curve $\varphi:I\rightarrow \cM$ is
singular at cosmological time $t_0\in I$ iff its complete lift
$\gamma=c(\varphi)$ meets the zero section of $T\cM$ for
$t=t_0$. Hence the embedded components of the orbit of $\varphi$ are
the $\pi$-projections of the set $\gamma(I)\cap {\dot T}\cM$, where
${\dot T}\cM$ is the slit tangent bundle of $\cM$ (which is defined as
the complement in $T\cM$ of the image of the zero section).

\paragraph{Future completeness when $\cM$ is compact.}

Since we assume that $\Phi$ is positive on $\cM$, each energy sublevel
set:
\be
\cM_E(C)\eqdef \{u\in T\cM~\vert~E(u)\leq C\}~~(C>0)
\ee
is contained in the tubular neighborhood of the zero section of $T\cM$
defined by the inequality $||u||_\cG\leq C\sqrt{2}$. Relation
\eqref{EnergyEq} implies that the complete lift $\dot{\varphi}$ of a
maximal cosmological curve $\varphi$ is contained for $t\geq t_0$
within the tubular neighborhood with $C=E_{\varphi}(t_0)$. When $\cM$
is compact, this tubular neighborhood is compact and this observation
together with the Escape Lemma implies that $\varphi(t)$ is defined
for all $t\geq t_0$, which shows that in this case the cosmological
flow is future-complete.

Notice, however, that $\cM$ is non-compact in most applications. When
$\cM$ is not compact, a maximal cosmological curve can ``escape to
infinity'' in the sense that its orbit for large times is not
contained in any compact subset of $\cM$; this is equivalent with the
statement that the curve has a Freudenthal end of $\cM$ among its
limit points. In this case, the late time behavior of those
cosmological curves which escape to infinity depends markedly on the
asymptotic form of $\Phi$ and $\cG$ near the Freudenthal ends of $\cM$
and the cosmological flow need not be future-complete.

\subsection{The universal similarity group}
\label{subsec:contsym}

\noindent Multifield cosmological models admit a universal two-parameter group
of similarities, which relate the cosmological curves of a model with
those of another model having the same target manifold $\cM$ but
different parameters $(M_0,\cG,\Phi)$. We first discuss scale
transformations of curves in $\cM$, which enters the definition of
this group action.

\begin{definition}
Let $\epsilon>0$. The {\em $\epsilon$-scale transform} of a curve
$\varphi:I\rightarrow \cM$ is the curve
$\varphi_\epsilon:I_\epsilon\rightarrow \cM$ defined through:
\be
I_\epsilon \eqdef \epsilon I=\{\epsilon t \vert t\in I\}
\ee
and:
\be
\varphi_\epsilon(t)\eqdef \varphi(t/\epsilon)~~\forall t\in I_\epsilon~~.
\ee
\end{definition}

\

\noindent The transformations $\varphi\rightarrow \varphi_\epsilon$
are called {\em scale transformations}. They define an action of the
multiplicative group $\R_{>0}$ on the set:
\be
C(\cM)\eqdef \sqcup_{I\in \Int}\cC^\infty(I,\cM)
\ee
of all smooth curves in $\cM$, where $\Int$ is the set of all
non-degenerate intervals on the real axis. Scale transformations are
not symmetries of the cosmological equation; the scale symmetry is
``broken'' by the scalar potential $\Phi$, being restored in the limit
$\Phi\rightarrow 0$ (which, as we will see below, corresponds to the
UV limit).

\

\noindent We consider the following similarities of the cosmological equation \eqref{eomsingle}:
\begin{itemize}
\item {\bf The parameter homothety.} The Lagrangian density $\cL$ of
\eqref{cL} and the action $S$ of \eqref{S} are homogeneous of degree
one under the {\em parameter homothety}:
\ben
\label{parhom}
\cG\rightarrow \lambda \cG~~,~~\Phi\rightarrow \lambda \Phi~~,~~M\rightarrow \lambda^{1/2} M ~~(\mathrm{thus}~~M_0\rightarrow \lambda^{1/2} M_0)~~,
\een
where $\lambda$ is a positive constant. As a consequence, the
equations of motion \eqref{EL} and the cosmological equation
\eqref{eomsingle} are invariant under such
transformations\footnote{Notice that the Levi-Civita connection
$\nabla$ of $\cG$ is invariant under constant rescalings of $\cG$.}.
\item {\bf The scale similarity.} 
The equations of motion \eqref{EL} are invariant under the transformations:
\be
t\rightarrow \epsilon t ~~,~~\Phi\rightarrow \Phi_\epsilon\eqdef \Phi/\epsilon^2~~
\ee
with $\epsilon>0$, which change $H=\frac{\dot{a}}{a}$ into
$\frac{1}{\epsilon}H$. Accordingly, the cosmological equation
\eqref{eomsingle} is invariant under:
\ben
\label{sim}
\varphi\rightarrow \varphi_\epsilon~~,~~\Phi\rightarrow \Phi_\epsilon\eqdef \Phi/\epsilon^2~~(\epsilon>0)~~.
\een
\end{itemize}

\noindent Let $\Met(\cM)$ be the set of all Riemannian metrics defined
on $\cM$ and $\Pot_+(\cM)\eqdef \cC^\infty(\cM,\R_{>0})$ be the set of
all positive smooth functions defined on $\cM$. The {\em universal cosmological
similarity group} is the multiplicative group $T=\R_{>0}\times \R_{>0}$, where
pairs of positive numbers multiply componentwise:
\be
(\lambda_1,\epsilon_1)(\lambda_2,\epsilon_2)\eqdef (\lambda_1\lambda_2,\epsilon_1\epsilon_2)~~.
\ee
The transformations above induce an action $\rho_{\param}$ of $T$
on the set $\Par(\cM)\eqdef \R_{>0}\times \Met(\cM)\times \Pot_+(\cM)$
of parameters of models with fixed target space $\cM$:
\ben
\label{rhoparam}
\rho_{\param}(\lambda,\epsilon)(M_0,\cG,\Phi)\eqdef (\lambda^{1/2}M_0,\lambda\cG, \frac{\lambda}{\epsilon^2} \Phi)~~\forall (\lambda,\epsilon)\in T~~\forall (M_0,\cG,\Phi)\in  \Par(\cM)~~.
\een
This action is free. On the other hand, scale transformations define
an action $\rho_0$ of $T$ on the space $C(\cM)$ of all smooth curves
in $\cM$:
\be
\rho_0(\lambda,\epsilon)(\varphi)\eqdef \varphi_\epsilon~~\forall (\lambda,\epsilon)\in T~~\forall \varphi\in C(\cM)~~. 
\ee

\begin{definition} The {\em universal similarity action} is the action
$\rho$ of $T$ on the set $C(\cM)\times \Par(\cM)$ given by:
\be
\rho(\lambda,\epsilon)(\varphi,M_0,\cG,\Phi)=(\rho_0(\lambda,\epsilon)(\varphi),
\rho_{\param}(\lambda,\epsilon)(M_0,\cG,\Phi))=(\varphi_\epsilon,\lambda^{1/2}M_0,\lambda\cG,
\frac{\lambda}{\epsilon^2} \Phi)~~.
\ee
\end{definition}

\noindent For any $(\lambda,\epsilon)\in T$, we have: 
\be
\rho_0(\lambda,\epsilon)(\cS^{M_0,\cG,\Phi}(\cM))=\cS^{\rho_\param(\lambda,\epsilon)(M_0,\cG,\Phi)}(\cM)~~,
\ee
where $\cS^{M_0,\cG,\Phi}(\cM)$ denotes the set of cosmological
curves of the model with target $\cM$ and parameters $(M_0,\cG,
\Phi)$. Thus $\varphi:I\rightarrow \cM$ is a cosmological curve for
the model with parameters $(M_0,\cG,\Phi)$ iff
$\varphi_\epsilon:I_\epsilon\rightarrow \cM$ is a cosmological curve
for the model with parameters $(\lambda^{1/2}M_0,\lambda\cG,
\frac{\lambda}{\epsilon^2} \Phi)$.

The universal similarity action allows one to eliminate the overall
scale of $\Phi$ and another scale from the problem. For example, one
can fix the overall scales of $\cG$ and $\Phi$, in which case one is
left with the single scale set by $M_0$. Equivalently, one can
eliminate the rescaled Planck mass by setting $M_0=1$ and the overall
scale of $\Phi$, in which case one is left with the single scale set
by $\cG$.  Notice that one cannot fix the scales of $M_0$ and $\cG$
independently.  Since $\Phi$ plays a distinguished role in this
regard, it is natural to consider the stabilizer $T_\ren\simeq
\R_{>0}$ of $\Phi$ in $T$ with respect to the universal similarity action:
\ben
\label{Tren}
T_\ren=\Stab_T(\Phi)=\{(\lambda,\epsilon)\in T\vert \lambda=\epsilon^2\}
\een
This is the subgroup of $T$ which can be used to rescale $M_0$ and
$\cG$ once the scale of $\Phi$ has been fixed; it is the {\em
  renormalization group} considered in Section \ref{sec:ren}.

\subsection{Cosmological conjugation and equivalence}
\label{subsec:conjequiv}

\noindent There exist two natural equivalence relations between
cosmological models, which are the isomorphism relations of two
underlying groupoids. To describe them, we first note some obvious
properties of the canonical lift $c:C(\cM)\rightarrow C_s(T\cM)$ of
smooth curves from a manifold $\cM$ to its tangent bundle. Consider
two manifolds $\cM_i$ with tangent bundle projections
$\pi_i:T\cM_i\rightarrow \cM_i$ ($i=1,2$) and canonical curve lifts
$c_i:C(\cM_i)\rightarrow C_s(T\cM_i)$. By definition, a {\em semispray
map} from $\cM_1$ to $\cM_2$ is a smooth map $f:T\cM_1\rightarrow
T\cM_2$ which satisfies $f_\ast(T^sT\cM_1)=T^sT\cM_2$. Precomposition
with a semispray map $f$ takes semispray curves in $T\cM_1$ to
semispray curves in $T\cM_2$, i.e. we have:
\be
f\circ (C_s(T\cM_1))=C_s(T\cM_2)~~.
\ee
Moreover, $f$ induces a map ${\hat f}:C(\cM_1)\rightarrow C(\cM_2)$
defined through:
\ben
\label{hatf}
{\hat f}(\varphi_1)\eqdef \pi_2 \circ f\circ c_1(\varphi_1)~~\forall \varphi_1\in C(\cM_1)~~.
\een
This determines the precomposition of semispray curves with $f$ in the
sense that two smooth curves $\varphi_i:I\rightarrow \cM_i$ ($i=1,2$)
satisfy $\varphi_2={\hat f}(\varphi_1)$ iff $c_2(\varphi_2)=f\circ
c_1(\varphi_1)$.

Given a third manifold $\cM_3$ with tangent bundle projection
$\pi_3:T\cM_3\rightarrow \cM_3$ and a second semispray map
$h:T\cM_2\rightarrow T\cM_3$, one easily checks the relation:
\be
\widehat{h\circ f}={\hat h}\circ {\hat f}~~.
\ee
Moreover, for any manifold $\cM$ we have
$\widehat{\id_{T\cM}}=\id_{C(\cM)}$.  Hence the maps $\cM\rightarrow
C(\cM)$ and $f\rightarrow {\hat f}$ define a functor from the category
of manifolds and semispray maps to the category of sets.

\paragraph{Cosmological conjugations.}

Notice that a cosmological flow curve represents the time evolution of
the state of a multifield cosmological model, where the tangent bundle
of $\cM$ is the space of classical states. Hence two cosmological
models parameterized by $\fM_i=(M_{0i},\cM_i,\cG_{i},\Phi_{i})$
($i=1,2$) and whose semisprays we denote by $S_i$ can be
identified if there exists a smooth semispray map $f:T\cM_1\rightarrow
T\cM_2$ (called {\em smooth cosmological conjugation}) which maps the
cosmological flow curves of the first model into those of the second,
i.e. $\gamma_1:I\rightarrow T\cM_1$ is a cosmological flow curve for
$\fM_1$ iff $f\circ\gamma_1:I\rightarrow T\cM_2$ is a cosmological
flow curve for $\fM_2$. This amounts to the requirement that $f$ is a
smooth topological conjugation between the cosmological flows of the
two models (see Appendix \ref{app:tsequiv}), which in turn is
equivalent with the condition:
\be
f_\sharp(S_1)=S_2~~,
\ee
where $f_\sharp$ denotes the $f$-pushforward of vector fields.  Since
every cosmological flow curve is the canonical lift of a cosmological
curve, a semispray map $f$ from $\cM_1$ to $\cM_2$ is a cosmological
conjugation iff the induced curve map ${\hat f}$ of \eqref{hatf} takes
cosmological curves of the model $\fM_1$ into those of the model
$\fM_2$. Since $f\rightarrow {\hat f}$ is a functor, it is clear that
cosmological models and smooth cosmological conjugations form a
groupoid, whose isomorphism classes we call {\em smooth cosmological
conjugacy classes}.

When two models $\fM_1$ and $\fM_2$ are isomorphic in this groupoid,
we write $\fM_1\equiv\fM_2$ and say that the models are {\em smoothly
conjugate}. It is clear that the equivalence relation $\equiv$ depends
only on the rescaled scalar triples of the models, since so do their
cosmological equations. In particular, a parameter homothety with
parameter $\lambda$ corresponds to a conjugation with $f=\id_{T\cM}$
between the model parameterized by $(M_0,\cM,\cG,\Phi)$ and that
parameterized by $(\lambda^{1/2}M_0,\cM,\lambda \cG,\lambda \Phi)$. If
one also requires $M_{01}=M_{02}$ then the conjugation is called {\em
strict}.  Any smooth cosmological conjugation is the composite of a
strict conjugation with a parameter homothety.

A strict smooth cosmological conjugation preserves basic on-shell
cosmological observables in the sense that the pullback map
$f^\ast:\cC^\infty(T\cM_2)\rightarrow \cC^\infty(T\cM_1)$ satisfies:
\be
\widehat{(f^\ast)(F_2)}_{\gamma_1}=({\hat F}_2)_{f\circ \gamma_1}
\ee
for any basic observable $F_2\in \cC^\infty(T\cM)$ of the second model
and any cosmological flow curve $\gamma_1$ of the first model.

\begin{eg}
A particularly simple class of strict smooth conjugations arises from
isomorphisms of scalar triples. We say that two scalar triples
$(\cM_1,\cG_{1},\Phi_{1})$ and $(\cM_2,\cG_{2},\Phi_{2})$ are {\em
isomorphic} if there exists an isometry
$f_0:(\cM_1,\cG_{1})\rightarrow (\cM_2,\cG_{2})$ such that
$\Phi_{1}=\Phi_{2}\circ f_0$; in this case, $f_0$ is called an {\em
isomorphism of scalar triples} and we write
$(\cM_1,\cG_{1},\Phi_{1})\simeq (\cM_2,\cG_{2},\Phi_{2})$.  We say
that the models parameterized by $\fM_1$ and $\fM_2$ are {\em
isomorphic} and write $\fM_1\simeq \fM_2$ if $M_{01}=M_{02}$ and
$(\cM_1,\cG_{1},\Phi_{1})\simeq (\cM_2,\cG_{2},\Phi_{2})$. The
differential $\dd f_0$ of an isomorphism $f_0$ of cosmological models
is a strict smooth conjugation.
\end{eg}

A (necessarily strict) smooth cosmological conjugation between a model
parameterized by $(M_0,\cM,\cG_0,\Phi_0)$ and itself is called a {\em
cosmological symmetry} of that model. Automorphisms of a cosmological
model are sometimes called {\em visible symmetries} while its
remaining symmetries are called {\em hidden} (see
\cite{Noether1,Noether2,Tim19,Hesse}).

\begin{remark}
The study of cosmological symmetries of multifield models in the
mathematical generality considered here has been limited. If one
restricts attention to Lie groups, this becomes the problem of
determining the Lie symmetries of the cosmological equation
\eqref{eomsingle} and of classifying rescaled scalar triples whose
cosmological equation admits given Lie groups of symmetries. Various
local results about Lie symmetries of multifield models can be found
in \cite{TP, Giacomini, P, Noether1, Mondal}. The literature contains
limited information on the global geometry of the resulting scalar
manifolds; see \cite{Noether2, Hesse} for some results in that
direction.
\end{remark}

\paragraph{Cosmological equivalences.}

An equivalence relation weaker than conjugation arises if one
identifies cosmological curves up to increasing reparameterization of
the cosmological time. We say that a smooth curve
$\varphi:I\rightarrow \cM$ is a {\em pre-cosmological curve} of the
model parameterized by $\fM=(M_0,\cM,\cG,\Phi)$ if there exists an
increasing reparameterization $\alpha:J\rightarrow I$ such that
$\varphi\circ \alpha$ is a cosmological curve of $\fM$. A
smooth semispray map $f:T\cM_1\rightarrow T\cM_2$ is called a {\em
smooth cosmological equivalence} between the models $\fM_1$ and $\fM_2$ if
the map ${\hat f}$ of \eqref{hatf} maps the pre-cosmological curves of
$\fM_1$ into those of $\fM_2$. This amounts to the requirement that
${\hat f}$ identifies the oriented cosmological {\em orbits} of the
two models. Notice that cosmological equivalences between $\fM_1$ and
$\fM_2$ differ\footnote{Indeed, a cosmological equivalence does {\em
not} identify the cosmological {\em flow} curves of the two models up
to reparameterization, but up to reparameterization combined with a
time-dependent rescaling within the tangent bundle fibers. The notion
of topological equivalence of dynamical systems is recalled in
Appendix \ref{app:tsequiv}.} from smooth {\em dynamical} equivalences,
which are defined by the requirement $f$ gives a smooth topological
equivalence between the cosmological flows of the models.

It is clear from the properties of ${\hat f}$ that cosmological models
and cosmological equivalence form a groupoid, whose isomorphism
classes we call {\em smooth cosmological equivalence classes}. When
$\fM_1$ and $\fM_2$ are isomorphic in this groupoid, we write
$\fM_1\sim\fM_2$ and say that the two models are {\em smoothly
equivalent}. This equivalence relation is weaker than cosmological
conjugation in the sense that $\fM_1\equiv\fM_2$ implies
$\fM_1\sim\fM_2$. Like cosmological conjugation, cosmological
equivalence depends only on the rescaled scalar triples; we say that
the equivalence is strict if $M_{01}=M_{02}$.

In Section \ref{sec:ren}, we will introduce other equivalence
relations, which arise from the study of UV and IR limits.

\subsection{Some remarks on the early and late time behavior of cosmological curves}
\label{subsec:limsets}

\noindent Consider a maximal cosmological curve
$\varphi:(a_-,a_+)\rightarrow \cM$, where $a_\pm\in \oR$ can be chosen
such that $a<0<b$ since the cosmological equation is
autonomous. Recall that the ordinary $\alpha$- and $\omega$- limit
sets of $\varphi$ are defined through \cite{Palis}:
\beqa
&&\Lim_\alpha\varphi\eqdef \{m\in \cM~|~\exists t_n\rightarrow a_+: \lim_{n\rightarrow \infty} \varphi(t_n)=m\}\subset \cM~~\nn\\
&&\Lim_\omega\varphi\eqdef \{m\in \cM~|~\exists t_n\rightarrow a_-: \lim_{n\rightarrow \infty} \varphi(t_n)=m\}\subset \cM~~.
\eeqa
Let us assume that $\cM$ is compact. Then for any $\omega$-limit point
$m$ of $\varphi$ and any sequence $t_n\in (a_-,a_+)$ with
$t_n\rightarrow a_+$ and $\varphi(t_n)\rightarrow m$, the tangent
vectors $\dot{\varphi}(t_n)$ are contained in a compact subset of
$T\cM$ since $E_\varphi(t_n)$ stays bounded by \eqref{EnergyEq}. Hence
there exists a subsequence $t'_n$ of $t_n$ such that
$\dot{\varphi}(t'_n)$ converges to a point $u\in T\cM$ with
$\pi(u)=m$. Since $u$ is a limit point of the integral curve
$\gamma=c(\varphi)$ of the vector field $S\in \cX(T\cM)$, we must have
$\cS(u)=0$ i.e. $u$ is a stationary point of the cosmological flow.
Hence $u$ belongs to the set $(\Crit \Phi)_0$ and $m=\pi(u)$ belongs
to $\Crit\Phi$. This shows that the ordinary $\omega$-limit points of
a cosmological curve are critical points of $\Phi$. Moreover, the set
$\Lim_\omega\varphi=\pi(\Lim_\omega c(\varphi))$ is nonempty, compact
and connected since the map $\pi$ is continuous and the set
$\Lim_\omega c(\varphi)$ has these properties by
\cite[Prop. 1.4]{Palis}.

When $\cM$ is not compact, both ordinary limit sets of a maximal
cosmological curve $\varphi$ may be empty. Indeed, the orbit of
$\varphi$ need not be contained in any compact subset of $\cM$, which
means that there there may exist a sequence of cosmological times
$t_n\in (a_-,a_+)$ such that $\varphi(t_n)$ approaches a Freudenthal
end of $\cM$ when $t_n\rightarrow a_-$ or $t_n\rightarrow a_+$. To
account for this, we consider the end compactification $\hcM$ of $\cM$
(which is a compact Hausdorff space containing $\cM$ as a dense
subset) and view $\varphi$ as the continuous curve $\varphi_e\eqdef
\iota\circ \varphi$ in the topological space $\hcM$, where
$\iota:\cM\hookrightarrow \hcM$ is the inclusion. By definition, the
{\em extended limit sets} of $\varphi$ are the $\alpha$- and
$\omega$-limit sets of $\varphi_e$ in the topological space $\hcM$:
\beqa
&&\Lim_\alpha^e\varphi\eqdef \Lim_\alpha\varphi_{e}=\{{\hat m}\in \hcM~|~\exists t_n\rightarrow a: \lim_{n\rightarrow \infty} \varphi(t_n)={\hat m}\}\subset \hcM~~\nn\\
&&\Lim_\omega^e\varphi\eqdef \Lim_\omega\varphi_{e}=\{{\hat m}\in \hcM~|~\exists t_n\rightarrow b: \lim_{n\rightarrow \infty} \varphi(t_n)={\hat m}\}\subset \hcM~~.
\eeqa
Let $\Ends(\cM)\eqdef \hcM\setminus \cM$ be the space of ends of $\cM$
(which is a totally disconnected topological space). We have:
\be
\Lim_\alpha^e\varphi=\Lim_\alpha\varphi \sqcup \Lambda_\alpha\varphi~~,~~\Lim_\omega^e\varphi=\Lim_\omega\varphi \sqcup \Lambda_\omega\varphi~~,
\ee
where:
\be
\Lambda_\alpha\varphi\eqdef \Lim_\alpha\varphi_e \cap \Ends(\cM)~~,~~\Lambda_\omega\varphi\eqdef \Lim_\omega\varphi_e \cap
\Ends(\cM)
\ee
are the sets of $\alpha$- and $\omega$- {\em limit ends} of
$\varphi$. With these definitions, $\varphi$ is contained in some
compact subset of $\cM$ iff
$\Lambda_\alpha\varphi=\Lambda_\omega\varphi=\emptyset$.  When the end
compactification $\hcM$ is sequentially compact, an argument similar
to that of \cite[Prop. 1.4]{Palis} shows that $\Lim_\alpha^e\varphi$
and $\Lim_\omega^e\varphi$ are nonempty, compact and connected subsets
of $\hcM$. The fact that a Freudenthal end of $\cM$ can act as a limit
point for a cosmological curve is a fundamental feature of models with
non-compact target space.

In general, the continuous map $\varphi_e:(a_-,a_+)\rightarrow \cM$
need not have limits for $t\rightarrow a_-$ or $t\rightarrow a_+$ in
$\hcM$. When the corresponding limit exists, we call it the $\alpha$-
or $\omega$- extended limit of $\varphi$ and denote it by:
\be
\alim\varphi\eqdef \lim_{t\rightarrow a_-}\varphi_e(t)\in \hcM~~\mathrm{respectively}~~\olim\varphi\eqdef \lim_{t\rightarrow a_+}\varphi_e(t)\in \hcM~~.
\ee

\section{Scaling limits and approximations}
\label{sec:ScalingLimits}

In this section, we study the UV and IR scaling limits of classical
multifield cosmological models with rescaled Planck mass $M_0$ and
scalar triple $(\cM,\cG,\Phi)$, showing that cosmological curves are
approximated in these limits by a reparameterization of the geodesic
flow of $(\cM,\cG)$ and by the gradient flow of $(\cM,\cG,V)$, where
$V=M_0\sqrt{2\Phi}$ is the {\em classical effective potential} of the
model. We also study the consistency conditions for these
approximations, showing that they differ from commonly used
approximations in cosmology, such as the slow roll approximation and
its slow roll-slow turn variant as well as from the gradient flow
approximation of \cite{genalpha}.

\subsection{The UV and IR limits}

\noindent By definition, the behavior of the scale transform
$\varphi_\epsilon$ at time $t$ recovers the behavior of $\varphi$ at
time $t/\epsilon$. Notice that $\varphi$ can be recovered from its
scale transform by setting $\epsilon=1$:
\be
\varphi=\varphi_1~~.
\ee
A time interval $\Delta t$ is rescaled to $\frac{1}{\epsilon}\Delta t$
under the scale transform by $\epsilon$. Hence cosmological time
intervals are compressed for large $\epsilon$ and expanded for small
$\epsilon$.

Intuitively, the limits of small and large $\epsilon$ capture the low
and high frequency components of $\varphi$ since the large $\epsilon$
limit sharpens the oscillations of $\varphi_\epsilon$ while the small
$\epsilon$ limit stretches them out. To make this quantitative,
consider for simplicity the case $\cM=\R^d$ and a cosmological curve
$\varphi$ which is defined on $\R$. Then
$\varphi(t)=(\varphi^1(t),\ldots, \varphi^d(t))$, where
$\varphi^i:\R\rightarrow \R$ (for $i=1,\ldots,d$) are the projections
of $\varphi$ on the Cartesian coordinate axes. Since in this case
$\varphi$ is a vector-valued function, it has a Fourier decomposition:
\be
\varphi(t)=\frac{1}{\sqrt{2\pi}} \int_\R \dd \omega e^{\i\omega t} \hat{\varphi}(\omega)~~,
\ee
where ${\hat \varphi}:\R\rightarrow \R^d$ is the Fourier transform of $\varphi$:
\be
{\hat \varphi}(\omega)=\frac{1}{\sqrt{2\pi}} \int_\R \dd t e^{-\i\omega t} \varphi(t)~~.
\ee
We have: 
\be
\varphi_\epsilon(t)=\varphi(t/\epsilon)=\frac{1}{\sqrt{2\pi}} \int_\R \dd \omega e^{\i\omega t/\epsilon} \hat{\varphi}(\omega)=
\frac{\epsilon}{\sqrt{2\pi}} \int_\R \dd \omega e^{\i\omega t} \hat{\varphi}(\epsilon \omega)~~,
\ee
where we performed the change of variables $\omega\rightarrow \epsilon
\omega$. Hence the Fourier components of $\varphi_\epsilon$ are:
\be
\widehat{\varphi_\epsilon}(\omega)=\epsilon ~\hat{\varphi}(\epsilon \omega)=\epsilon ~({\hat \varphi})_{1/\epsilon}(\omega)
\ee
and its Fourier transform is given by
$\widehat{\varphi_\epsilon}=\epsilon ~({\hat
\varphi})_{1/\epsilon}$. When $\epsilon$ is very large or very small,
the Fourier component of $\varphi_\epsilon$ at frequency $\omega$
reproduces up to a rescaling of the amplitude the Fourier component of
$\varphi$ at the much higher (respectively much lower) frequency
$\epsilon \omega$. Hence the {\em UV limit} $\epsilon\rightarrow
\infty$ corresponds to the high frequency behavior of $\varphi$ while
the {\em IR limit} $\epsilon\rightarrow 0$ corresponds to its low
frequency behavior; these {\em scaling limits} correspond to the
ultraviolet and infrared limits of the oscillation spectrum of
$\varphi$. Notice that the kinetic energy of a plane wave component
${\hat \varphi}(\omega)e^{\i \omega t}$ of $\varphi$ equals
$E^k_\varphi(\omega)\eqdef \frac{1}{2}\omega^2||{\hat
\varphi}(\omega)||_\cG^2$ and we have:
\be
E^k_{\varphi_\epsilon}(\omega)=\frac{1}{2}\omega^2 ||\hat{\varphi_\epsilon}(\omega)||_\cG^2=\frac{\epsilon^2}{2}\omega^2 ||{\hat \varphi}(\epsilon \omega)||_\cG^2=\epsilon^2 E^k_\varphi(\epsilon \omega)~~.
\ee

\begin{remark} Recall that the generalized curvatures of a curve in a
Riemannian manifold do not depend on parameterization and hence are
properties of the image of that curve rather than of the curve
itself. Hence the generalized curvatures of $\varphi_\epsilon$ and
$\varphi$ at every point on the orbit
$\im(\varphi_\epsilon)=\im(\varphi)$ coincide. Since scale transforms
do not change the generalized curvatures of $\varphi$, the limits of
small and large $\epsilon$ have no connection with slow or fast turns
of the cosmological orbit, which form the focus of other approximation
schemes in cosmology.
\end{remark}

\subsection{The rescaled cosmological equation}

\noindent For any $n\in \N_{>0}$, we have:
\ben
\label{derc}
\frac{\dd^n \varphi_\epsilon}{\dd t^n}=\frac{1}{\epsilon^n} \left(\frac{\dd^n \varphi}{\dd t^n}\right)_{\!\epsilon}~~
\een
and:
\ben
\label{nablac}
\nabla_t^n \frac{\dd \varphi_\epsilon}{\dd  t}=\frac{1}{\epsilon^{n+1}} \left(\nabla_t^n\frac{\dd \varphi}{\dd t}\right)_{\!\epsilon}~~.
\een
These relations imply that a curve $\varphi:I\rightarrow \cM$
satisfies the cosmological equation \eqref{eomsingle} iff its scale
transform $\varphi_\epsilon$ satisfies the {\em $\epsilon$-rescaled
cosmological equation}:
\ben
\label{eomdef}
\epsilon^2 \nabla_t \frac{\dd \varphi_\epsilon(t)}{\dd t}\!+\!\frac{\epsilon}{M_0} 
\left[\epsilon^2\Bvert\frac{\dd \varphi_\epsilon(t)}{\dd
t}\Bvert_\cG^2+2\Phi(\varphi_\epsilon(t))\right]^{1/2}\!\!\frac{\dd
\varphi_\epsilon(t)}{\dd t}+ (\grad_{\cG} \Phi)(\varphi_\epsilon(t))=0~~.
\een
Upon dividing by $\epsilon^2$, the latter is equivalent with:
\ben
\label{eomdef0}
\nabla_t \frac{\dd \varphi_\epsilon(t)}{\dd t}+\frac{1}{M_0}
\left[\Bvert\frac{\dd \varphi_\epsilon(t)}{\dd
t}\Bvert_\cG^2+2\Phi_\epsilon(\varphi_\epsilon(t))\right]^{1/2}\frac{\dd
\varphi_\epsilon(t)}{\dd t}+ (\grad_{\cG}
\Phi_\epsilon)(\varphi_\epsilon(t))=0~~,
\een
where $\Phi_\epsilon\eqdef \Phi/\epsilon^2$. Hence $\varphi$ satisfies
the cosmological equation of the scalar triple $(\cM,\cG,\Phi)$ iff
$\varphi_\epsilon$ satisfies the cosmological equation of the scalar
triple $(\cM,\cG,\Phi_\epsilon)$. This also follows from the
fact that \eqref{sim} is a symmetry of the cosmological
equation. In particular, the UV limit $\epsilon\rightarrow \infty$
amounts to taking the overall scale of $\Phi$ to zero, while the IR
limit $\epsilon\rightarrow 0$ amounts to taking the overall scale of
$\Phi$ to infinity.

The semispray defined by \eqref{eomdef0} is:
\be
S_\epsilon=\mS-\cH_\epsilon C-\frac{1}{\epsilon^2}(\grad_\cG\Phi)^v~~,
\ee
where the function $\cH_\epsilon:T\cM\rightarrow \R_{>0}$ is defined through:
\be
\cH_\epsilon(u)=\frac{1}{M_0}\left[||u||_\cG^2+\frac{2\Phi(\pi(u))}{\epsilon^2}\right]^{1/2}~~.
\ee
In the UV limit, $S_\epsilon$ reduces to the $N_0$-modification
of the geodesic spray of the rescaled scalar manifold $(\cM,\cG_0)$:
\ben
\label{SUV}
S_{\UV}\eqdef \mS-N_0 C~~,
\een
where $N_0:T\cM\rightarrow \R_{\geq 0}$ is the norm function defined by
$\cG_0\eqdef \frac{1}{M_0^2}\cG$ on $T\cM$:
\be
N_0(u)\eqdef ||u||_{\cG_0}=\frac{1}{M_0}||u||_\cG~~.
\ee
Notice that the geodesic spray $\mS$ of $(\cM,\cG)$ coincides with
that of $(\cM,\cG_0)$ since it is invariant under constant rescalings
of $\cG$.  Also notice that $N_0$ and the spray \eqref{SUV} are
continuous on $T\cM$ but smooth only on the slit tangent bundle
$\dot{T}\cM$. In the IR limit $\epsilon\rightarrow 0$, $S_\epsilon$ is
approximated by the vertical vector field:
\be
S^{\IR}_\epsilon=-\frac{1}{\epsilon}\frac{1}{M_0}\sqrt{2\Phi^v}C-\frac{1}{\epsilon^2}(\grad_\cG \Phi)^v=-\frac{1}{\epsilon^2 M_0}\sqrt{2\Phi^v} S_{\IR,\epsilon}~~,
\ee
where: 
\ben
\label{SIReps}
S_{\IR,\epsilon}\eqdef \epsilon C-(\grad_\cG V)^v \in \cX(T\cM)
\een
and we defined the {\em classical effective scalar potential} $V$
through:
\ben
\label{Vdef}
V\eqdef M_0\sqrt{2\Phi}~~.
\een
Notice that $J(S^{\IR}_\epsilon)=0$, hence $S^{\IR}_\epsilon$ is no longer a
semispray. This signals that the IR limit is degenerate. Here $\Phi^v$
is the vertical lift of $\Phi$, which is defined through:
\be
\Phi^v(u)=\Phi(\pi(u))~~\forall u\in T\cM~~.
\ee

\subsection{The UV and IR expansions}

\noindent The rescaled cosmological equation can be used to construct
{\em UV and IR expansions} of cosmological curves and of the
cosmological flow. When $\epsilon$ is large, one can seek solutions
$\varphi_\epsilon$ to \eqref{eomdef0} which are expanded in positive
powers of $\frac{1}{\epsilon^2}$; then $\varphi(t)\eqdef
\varphi_\epsilon(\epsilon t)$ is a solution of the cosmological
equation \eqref{eomsingle} which is expanded in non-negative powers of
$\Phi$. This amounts to treating $\Phi$ as small, Taylor expanding the
reduced Hubble function \eqref{cH} as:
{\footnotesize \ben
\label{cHUV}
\cH(u)=||u||_\cG\left[1+\frac{2\Phi(\pi(u))}{||u||_\cG^2}\right]^{1/2}=||u||_{\cG}\left[1+\frac{\Phi(\pi(u))}{||u||_\cG^2}-
  \frac{1}{2}\left(\frac{\Phi(\pi(u))}{||u||_\cG^2}\right)^2+\ldots\right]
\een}
\!\!and seeking solutions $\varphi$ of the cosmological equation
expanded in powers of $\Phi$. This produces the {\em UV expansion} of
cosmological curves and a corresponding expansion of the cosmological
flow. Substituting \eqref{cHUV} into \eqref{Q} gives an expansion
of the cosmological semispray in powers of $\Phi$.

When $\epsilon$ is small, one can seek solutions
$\varphi_\epsilon$ of \eqref{eomdef} which are expanded in powers of
$\epsilon$; then $\varphi(t)\eqdef \varphi_\epsilon(\epsilon t)$ is a
solution of \eqref{eomsingle} expanded in powers of
$\frac{1}{\sqrt{2\Phi}}$. This amounts to treating $\Phi$ as large
and expanding the reduced Hubble function as:
{\scriptsize \ben
\label{cHIR}
\cH(u)=\sqrt{2\Phi(\pi(u))}\left[1+\left(\frac{||u||_\cG}{\sqrt{2\Phi(\pi(u))}}\right)^2\right]^{1/2}=
\sqrt{2\Phi(\pi(u))}\left[1+\frac{1}{2}\left(\frac{||u||_\cG}{\sqrt{2\Phi(\pi(u))}}\right)^2-
  \frac{1}{8}\left(\frac{||u||_\cG}{\sqrt{2\Phi(\pi(u))}}\right)^4+\ldots\right]
\een}
\!\!Substituting \eqref{cHIR} into \eqref{Q} produces an expansion of
the cosmological semispray in powers of
$\frac{1}{\sqrt{2\Phi}}$. Substituting \eqref{cHIR} into the
cosmological equation and dividing both sides by
$\sqrt{2\Phi(\varphi(t))}$ brings \eqref{eomsingle} to the form:
{\scriptsize \ben
\label{IRPhiExp}
\frac{1}{\sqrt{2\Phi(\varphi(t))}}\nabla_t
\dot{\varphi}(t)+\frac{1}{M_0}
\left[1+\left(\frac{||\dot{\varphi}(t)||_\cG^2}{\sqrt{2\Phi(\varphi(t))}}\right)^2-\frac{1}{8}\left(\frac{||\dot{\varphi}(t)||_\cG}{\sqrt{2\Phi(\varphi(t))}}\right)^4+\ldots \right]^{1/2}\dot{\varphi}(t)+
(\grad_{\cG} \sqrt{2 \Phi})(\varphi(t))=0
\een}
\!\!and one can seek solutions expanded in powers of
$\frac{1}{\sqrt{2\Phi}}$. This produces an asymptotic expansion of
cosmological curves called the {\em IR expansion} and a corresponding
expansion of the cosmological flow. Notice that the small expansion
parameter multiplies the highest order term in \eqref{IRPhiExp}. As a
consequence, the first order approximant $\varphi_{\IR}$ is obtained
by solving a {\em first order} ODE, which means that the corresponding
approximant $\Pi_{\IR}$ to the cosmological flow $\Pi$ is not defined
on the whole tangent bundle $T\cM$ but on a closed submanifold of the
latter. It is only higher order approximants of the cosmological flow
which are defined on the entire tangent bundle of $\cM$. In the next
subsections, we discuss the first order UV and IR approximants of
cosmological curves and of the cosmological flow.

\begin{remark}
The UV and IR expansions admit a geometric description obtained by
writing $\varphi_\epsilon$ as:
\ben
\label{UVexp}
\varphi_\epsilon(t)=\exp_{\varphi_{\UV}(t/\epsilon)}\left(\sum_{n\geq 1} \frac{1}{\epsilon^{2n}} v_n(t)\right)~~\mathrm{with}~~v_n(t)\in T_{\varphi_{\UV}(t/\epsilon)}\cM
\een
respectively:
\ben
\label{IRexp}
\varphi_\epsilon(t)=\exp_{\varphi_{\IR}(t/\epsilon)}\left(\sum_{n\geq 1} \epsilon^{n} w_n(t)\right)~~\mathrm{with}~~w_n(t)\in T_{\varphi_{\IR}(t/\epsilon)}\cM~~,
\een
where $\varphi_{\UV}$ and $\varphi_{\IR}$ are the first order UV and IR
approximants of $\varphi$ discussed below and
$\exp_m:T_m\cM\rightarrow \cM$ denotes the exponential map of the
Riemannian manifold $(\cM,\cG_0)$ at the point $m\in \cM$. The
vector-valued functions $v_n(t)$ and $w_n(t)$ can be determined by
substituting \eqref{UVexp} and \eqref{IRexp} into the rescaled
cosmological equation and expanding respectively in powers of
$1/\epsilon$ or $\epsilon$; we will explain the technical details of
this procedure in a different publication. When $\epsilon$ is large
(respectively small), the arguments of the exponential maps in the
expressions above tend to zero and hence the right hand sides reduce
to the UV and IR approximants.
\end{remark}

\subsection{The first order UV and IR approximants}

\noindent Let $\varphi:I=(a_-,a_+)\rightarrow \cM$ be a non-constant
maximal cosmological curve with $0\in (a_-,a_+)$, where $a_-\in \R\cup
\{-\infty\}$ and $a_+\in \R\cup \{+\infty\}$.

\paragraph{The strict UV limit and first order UV approximant.}

Let us fix $t\in I_\epsilon=(\epsilon a_-, \epsilon a_+)$ and set
$\lambda\eqdef \sign(t)\in \{-1,0,1\}$. In the strict UV limit
$\epsilon\rightarrow \infty$, we have $\lim_{\epsilon\rightarrow
\infty}I_\epsilon=\R$ and:
\be
\lim_{\epsilon\rightarrow \infty} \varphi_\epsilon(t)=\lim_{\epsilon\rightarrow \infty}\varphi(t/\epsilon)=\varphi(0)~~,~~
\lim_{\epsilon\rightarrow \infty} \frac{\dd\varphi_\epsilon(t)}{\dd t}=\lim_{\epsilon\rightarrow \infty}\frac{1}{\epsilon}\dot{\varphi}(t/\epsilon)=0~~.
\ee
Hence the limiting curve $\varphi_\infty$ is the constant curve
defined on $I_\infty=\R$ whose image is the point $\varphi(0)$.

When $\epsilon$ is large but finite, the rescaled cosmological
equation \eqref{eomdef0} becomes:
\be
\nabla_t \frac{\dd \varphi_\epsilon(t)}{\dd t}+\frac{1}{M_0} \Bvert\frac{\dd \varphi_\epsilon(t)}{\dd t}\Bvert_\cG\frac{\dd \varphi_\epsilon(t)}{\dd t}=\O(1/\epsilon^2)
\ee
and we have:
\be
\varphi_\epsilon(t)=\varphi(t/\epsilon)=\varphi(0)+\O(1/\epsilon)~~,~~\frac{\dd\varphi_\epsilon(t)}{\dd t}=\frac{1}{\epsilon}\dot{\varphi}(t/\epsilon)=\frac{1}{\epsilon}\dot{\varphi}(0)+\O(1/\epsilon^2)
\ee
for $|t|\ll \epsilon$. Thus $\varphi_\epsilon$ is approximated up to
first order in $1/\epsilon$ by a curve $t\rightarrow
\varphi_{\UV,\epsilon}(t)$ which satisfies:
\ben
\label{eomUVeps}
\nabla_t \frac{\dd \varphi_{\UV,\epsilon}(t)}{\dd t}+\frac{1}{M_0}\Bvert\frac{\dd \varphi_{\UV,\epsilon}(t)}{\dd t}\Bvert_\cG\frac{\dd \varphi_{\UV,\epsilon}(t)}{\dd t}=0
\een
and:
\be
\varphi_{\UV,\epsilon}(0)=\varphi(0)~~,~~\frac{\dd \varphi_{\UV,\epsilon}}{\dd t}\Big{\vert}_{t=0}=\frac{1}{\epsilon}\dot{\varphi}(0)~~.
\ee
Since \eqref{eomUVeps} is scale-invariant, the curve
$\varphi_{\UV}:\R\rightarrow \cM$ defined through $\varphi_{\UV}(t)\eqdef
\varphi_{\UV,\epsilon}(\epsilon t)$ satisfies the same equation:
\ben
\label{eomUV}
\nabla_t \frac{\dd \varphi_{\UV}(t)}{\dd t}+\frac{1}{M_0}\Bvert\frac{\dd \varphi_{\UV}(t)}{\dd t}\Bvert_\cG\frac{\dd \varphi_{\UV}(t)}{\dd t}=0
\een
and the conditions:
\ben
\label{UVin}
\varphi_{\UV}(0)=\varphi(0)~~,~~\frac{\dd \varphi_{\UV}}{\dd t}\Big{\vert}_{t=0}=\dot{\varphi}(0)~~.
\een
It follows that $\varphi(t)=\varphi_\epsilon(\epsilon t)$ is
approximated by $\varphi_{\UV}(t)$ for $|t|\ll 1$. Notice that
\eqref{SUV} is the spray defined by equation \eqref{eomUV}.
The complete lift of maximal solutions to \eqref{eomUV} defines the flow $\Pi_{\UV}$ of the
$N_0$-modification $S_{\UV}$ of the geodesic spray $\mS$. 

An increasing reparameterization $\sigma\rightarrow \sigma(t)$ shows
that \eqref{eomUV} is equivalent with the geodesic equation:
\ben
\label{geodesiceq}
\nabla_{\sigma} \frac{\dd \varphi_{\UV}}{\dd \sigma}=0~~,
\een
of $(\cM,\cG)$, where the affine parameter $\sigma$ is obtained by
solving the second order ODE:
\ben
\label{sigmaeq}
\ddot{\sigma}(t)+\frac{1}{M_0} ||{\varphi}'_{\UV}(\sigma)||_\cG\dot{\sigma}(t)^2=0~~.
\een
Here the prime indicates derivation with respect to $\sigma$. Since
the Riemannian manifold $(\cM,\cG)$ is complete, its geodesic flow is
complete by the Hopf-Rinow theorem. As a consequence, the maximal
geodesic $\sigma\rightarrow \varphi_{\UV}(\sigma)$ is defined on the
entire real axis. Equation \eqref{sigmaeq} can be written as:
\be
\frac{\dd }{\dd t}\left(\frac{1}{\dot{\sigma}}\right)=\frac{1}{M_0}||\varphi'_{\UV}(\sigma)||_\cG~~.
\ee
Writing the left hand side as $\dot{\sigma}\frac{\dd}{\dd
\sigma}\left(\frac{\dd t}{\dd \sigma}\right)=\dot{\sigma}\frac{\dd^2
t}{\dd \sigma^2}$, this is equivalent with:
\be
\frac{\dd^2 t}{\dd \sigma^2}=\frac{1}{M_0}||\varphi'_{\UV}(\sigma)||_\cG\frac{\dd t}{\dd \sigma}\Longleftrightarrow \frac{\dd u}{\dd \sigma}=\frac{1}{M_0}||\varphi'_{\UV}(\sigma)||_\cG u~~,
\ee
where we set $u\eqdef \frac{\dd t}{\dd \sigma}$. Thus
$u(\sigma)=Be^{\frac{1}{M_0}\int^{\sigma}_0 \dd \sigma'
||\varphi'_{\UV}(\sigma')||_\cG}$ and:
\ben
\label{tsigma}
t(\sigma)=A+B\int_0^{\sigma} \dd \sigma' e^{\frac{1}{M_0}\int_0^{\sigma'} \dd \sigma'' ||\varphi'_{\UV}(\sigma'')||_\cG}~~,
\een
where $A$ and $B$ are integration constants. Since $\sigma$ is an
increasing parameter for $\varphi_{\UV}$ (i.e. we require
$\dot{\sigma}>0$ and hence $\frac{\dd t}{\dd \sigma}>0$), we must take
$B>0$. Notice that \eqref{sigmaeq} is invariant under affine
transformations of $\sigma$, which correspond to affine
transformations of $t$, i.e. to changing the constants $A$ and $B$ in
\eqref{tsigma}. Since we require $\dot{\sigma}>0$, only affine
transformations of the form $\sigma\rightarrow \alpha+\beta\sigma$
with $\beta>0$ are allowed. Using such transformations, we can set
$\sigma\vert_{t=0}=0$, which amounts to taking $A=0$ in
\eqref{tsigma}. Suppose that $\dot{\varphi}(0)\neq 0$, so that the
maximal geodesic $\varphi_{\UV}$ is not constant. Then the remaining
freedom of changing $B$ allows us to take $\sigma$ to be the proper
length parameter:
\be
s(t)=\int_{0}^t\dd t' ||\dot{\varphi}_{\UV}(t')||_\cG~~,
\ee
in which case equation \eqref{tsigma} (with $A=0$) gives:
\ben
\label{ts0}
t(s)=B\int_0^{s} \dd s' e^{\frac{s'}{M_0}}=B_0 \left[e^{\frac{s}{M_0}}-1\right]~~,
\een
where $B_0\eqdef M_0B$ and we used the relation $\Bvert\frac{\dd
\varphi_{\UV}(s)}{\dd s}\Bvert_\cG=1$. Moreover, \eqref{geodesiceq} becomes:
\ben
\label{geodesiceqs}
\nabla_{s} \frac{\dd \varphi_{\UV}}{\dd s}=0~~
\een
while conditions \eqref{UVin} require
$B_0=\frac{M_0}{||\dot{\varphi}(0)||_\cG}$ and:
\be
\varphi_{\UV}\vert_{s=0}=\varphi(0)~~,~~\frac{\dd\varphi_{\UV}}{\dd s}\bvert_{s=0}=\frac{\dot{\varphi}(0)}{||\dot{\varphi}(0)||_\cG}~~.
\ee
To arrive at the last relations, we noticed that
$\frac{\varphi_{\UV}}{\dd
s}\bvert_{s=0}=\dot{\varphi}_{\UV}(0)t'(s)\vert_{s=0}=\frac{B_0}{M_0}\dot{\varphi}(0)$,
where we used the second of conditions \eqref{UVin}. Since
$||\frac{\dd\varphi_{\UV}}{\dd s}||_\cG=1$, this requires
$B_0=\frac{M_0}{||\dot{\varphi}(0)||_\cG}$. In particular, \eqref{ts0}
reads:
\ben
\label{ts}
t(s)=\frac{M_0}{||\dot{\varphi}(0)||_\cG} \left[e^{\frac{s}{M_0}}-1\right]~~.
\een
Notice that \eqref{ts} requires $t(s)>
-\frac{M_0}{||\dot{\varphi}(0)||_\cG}$, so the maximal curve
$t\rightarrow \varphi_{\UV}(t)$ is defined on the interval
$\left(-\frac{M_0}{||\dot{\varphi}(0)||_\cG},+\infty\right)$.  In
particular, the first order UV approximation must break down if $t\in I$
is smaller than $-\frac{M_0}{||\dot{\varphi}(0)||_\cG}$. In practice,
the approximation becomes inaccurate for negative times which are
smaller in absolute value since it is only appropriate for small
$\frac{|s|}{M_0}$. We have $t(s)\approx
\frac{1}{||\dot{\varphi}(0)||_\cG}s$ when $\frac{|s|}{M_0}\ll 1$.

If $\dot{\varphi}(0)=0$, then $\varphi_{\UV}$ is the constant geodesic
at $\varphi(0)$ since it satisfies
$\varphi_{\UV}\vert_{\sigma=0}=\varphi(0)$ and
$\frac{\varphi_{\UV}}{\dd \sigma}\Big{\vert}_{\sigma=0}=0$. In this
case, $\varphi_\UV$ satisfies \eqref{eomUV} with respect to any
parameter $t$. Notice that relation \eqref{tsigma} for $A=0$ formally
gives $t=B\sigma$ with $B>0$.

For any non-zero vector $u$ of $T\cM$ such that $m\eqdef \pi(u)\in
\cM_0=\cM\setminus \Crit V$, there exists a unique maximal
cosmological curve $\varphi_u:I_u\rightarrow \cM$ (defined on an open
interval $I_u$ which contains zero) which satisfies $\varphi_u(0)=m$
and $\dot{\varphi}_u(0)=u$. Setting $n\eqdef \frac{u}{||u||_\cG}$,
this cosmological curve is approximated in the IR limit by the
reparameterization \eqref{ts} of the unique maximal normalized
geodesic $\psi_n:\R\rightarrow \cM$ of $(\cM,\cG)$ which satisfies
$\psi_n\vert_{s=0}=m$ and $\frac{\dd \psi_n}{\dd
s}\Big{\vert}_{s=0}=n$. Hence the cosmological time along $\varphi_u$
is recovered in this approximation as:
\be
t(s)=\frac{M_0}{||u||_\cG}\left[e^{\frac{s}{M_0}}-1\right]~~.
\ee
Under the parameter homothety \eqref{parhom}, the proper length
parameter $s$ of $\psi_u$ changes as $s\rightarrow \lambda^{1/2} s$
while the norm of $u$ changes as $||u||_\cG\rightarrow
\lambda^{1/2}||v||_\cG$. Thus we can use that similarity of the
cosmological equation to set $M_0=1$; this absorbs $M_0$ into $\cG$
and changes the normalization of the geodesic flow without changing
the cosmological time $t$. Hence the first order of the UV
approximation of the cosmological flow of the model parameterized by
$(M_0,\cM,\cG,\Phi)$ is determined by the normalized geodesic flow of
the rescaled scalar field metric $\cG_0\eqdef \frac{1}{M_0^2}\cG$,
which has proper length parameter $s_0=\frac{1}{M_0}s$ and norm
$||~||_{\cG_0}=\frac{1}{M_0}||~||_{\cG}$. Summarizing the discussion
above gives the following:

\begin{prop}
Consider the cosmological model parameterized by $(M_0,\cM,\cG,\Phi)$
and set $\cG_0\eqdef \frac{1}{M_0^2}\cG$. Let:
\be
\o(\cM)=\{0_m\vert m\in \cM\}\subset T\cM
\ee
be the image of the zero section of $T\cM$ and:
\be
{\dot T}\cM\eqdef T\cM\setminus \o(\cM)
\ee
be the slit tangent bundle of $\cM$.  For each $u\in T\cM$, let
$\varphi_u:I\rightarrow \cM$ be the maximal cosmological curve of this
model which satisfies:
\be
\varphi_u(0)=\pi(u)~~\mathrm{and}~~\dot{\varphi}_u(0)=u~~.
\ee
When $u\in {\dot T}\cM$, this curve is approximated in
the UV limit (for $t\in I$ with $t> -\frac{1}{||u||_{\cG_0}}$) by a
reparameterization of the maximal normalized geodesic
$\psi_n:\R\rightarrow \cM$ of the rescaled scalar manifold
$(\cM,\cG_0)$ which satisfies:
\be
\psi_n\vert_{s_0=0}=\pi(u)~~,~~\frac{\dd\psi_n}{\dd s_0}\Big{\vert}_{s_0=0}=n\eqdef \frac{u}{~~||u||_{\cG_0}}~~,
\ee
where the cosmological time $t$ is related to the proper length
parameter $s_0$ of $\psi_n$ through:
\be
t(s_0)=\frac{e^{s_0}-1}{~~||u||_{\cG_0}}~~.
\ee
When $u\in \o(\cM)$, the cosmological curve $\varphi_u$ is
approximated in the UV limit by the constant geodesic of $(\cM,\cG_0)$
at $\pi(u)$. Accordingly, the cosmological flow $\Pi:\cD\rightarrow
T\cM$ of the model parameterized by $(M_0,\cM,\cG,\Phi)$ is
approximated in the UV limit by the flow $\Pi_{\cG_0}:\cD_{\cG_0}
\rightarrow T\cM$ of the $N_0$-modification \eqref{SUV} of the
geodesic spray of $(\cM,\cG_0)$, which has maximal domain of definition:
\be
\cD_{\cG_0}=\Big\{(t,u)\in \R\times T\cM \big{\vert} t> -\frac{1}{||u||_{\cG_0}}\Big\}~~.
\ee
\end{prop}

\noindent Notice that the approximations in the proposition are
only asymptotic. 

\paragraph{The strict IR limit and first order IR approximant.}

When $\epsilon$ is small but not zero, the rescaled cosmological
equation \eqref{eomdef} gives:
\ben
\label{IReps}
\frac{1}{M_0} \epsilon \sqrt{2\Phi(\varphi_{\epsilon}(t))}\frac{\dd \varphi_{\epsilon}(t)}{\dd t}+ (\grad_{\cG} \Phi)(\varphi_{\epsilon}(t))=\O(\epsilon^2)~~
\een
which is equivalent with:
\ben
\label{eomIR}
\epsilon\frac{\dd \varphi_{\epsilon}(t)}{\dd t}+ (\grad_{\cG} V)(\varphi_{\epsilon}(t))=\O(\epsilon^2)~~,
\een
where the classical effective scalar potential $V$ was defined in
\eqref{Vdef}.  Hence $\varphi_\epsilon(t)$ is approximated to first
order in $\epsilon$ by the solution $\varphi_{\IR,\epsilon}(t)$ of the
equation:
\ben
\label{eomIReps}
\epsilon\frac{\dd \varphi_{\IR,\epsilon}(t)}{\dd t}+ (\grad_{\cG} V)(\varphi_{\IR, \epsilon}(t))=0
\een
which satisfies:
\be
\varphi_{\IR,\epsilon}(0)=\varphi(0)~~.
\ee
It follows that $\varphi(t)=\varphi_\epsilon(\epsilon t)$ is
approximated by the solution $\varphi_\IR(t)\eqdef
\varphi_{\IR,\epsilon}(\epsilon t)$ of the equation:
\ben
\label{Vgf}
\frac{\dd \varphi_{\IR}(t)}{\dd t}+ (\grad_{\cG} V)(\varphi_{\IR}(t))=0
\een
which satisfies:
\ben
\label{IRincond}
\varphi_{\IR}(0)=\varphi(0)~~.
\een
Notice that the first order approximant $\varphi_\IR$ of $\varphi$ is
entirely determined by $\varphi(0)$ and does not depend on
$\dot{\varphi}(0)$; one must consider higher orders of the IR
expansion in order to obtain an approximant of $\varphi$ which also
depends on $\dot{\varphi}(0)$. In the first order IR approximation,
the initial speed $\dot{\varphi}(0)\in T_{\varphi(0)}\cM$ is
approximated by the vector $-(\grad_\cG V)(\varphi(0))\in
T_{\varphi(0)}\cM$; this approximation can only be accurate when
$||\dot{\varphi}(0)+(\grad_\cG V)(\varphi(0))||_\cG$ is small.

Equation \eqref{Vgf} is equivalent with the condition:
\ben
\label{SIRcond}
S_{\IR}(\dot{\varphi}_{\IR}(t))=0~~,
\een
where:
\ben
\label{SIR}
S_{\IR}\eqdef -C-(\grad_\cG V)^v\in \cX(T\cM)
\een
is obtained from \eqref{SIReps} for $\epsilon=1$. Condition
\eqref{SIRcond} confines the canonical lift of $\varphi_\IR$ to the
{\em gradient flow shell} $\Grad_\cG V\subset T\cM$ of the effective
scalar triple $(\cM,\cG,V)$, which is defined as the graph of the
vector field $-\grad_\cG V$:
\be
\Grad_\cG V\eqdef \graph(-\grad_\cG V)=\{u\in T\cM\vert u=-(\grad_\cG V)(\pi(u))\}=\{u\in T\cM\vert S_{\IR}(u)=0\}~~.
\ee
To first order of the IR expansion, the cosmological flow curves
of the model degenerate to the canonical lifts of the gradient flow
curves of $(\cM,\cG,V)$, which lie within the gradient flow shell
$\Grad_\cG V$. More precisely, consider a point $u\in T\cM$ and set
$\pi(u)=m\in \cM$. Then the maximal cosmological curve $\varphi_u$ of
this model which satisfies:
\be
\varphi_u(0)=m~~\mathrm{and}~~\dot{\varphi}_u(0)=u
\ee
is approximated to first order of the IR expansion by the gradient
flow curve $\eta_m$ of $(\cM,\cG,V)$ which satisfies:
\be
\eta_m(0)=m
\ee
This approximation is good for small $|t|$ only when $||u+(\grad_\cG
V)(\pi(u))||_\cG$ is small, i.e. when $u$ is close the gradient flow
shell $\Grad_\cG V$ and it is most precise when $u\in \Grad_\cG V$.

In this approximation, a basic cosmological observable $F:T\cM\rightarrow \R$
reduces to the function:
\be
F_\IR\eqdef F\circ (\grad_\cG V )\in \cC^\infty(\cM)
\ee
defined on $\cM$. For any cosmological curve $\varphi$, we have:
\be
F_\varphi\approx_{\IR} F_\IR\circ \varphi_\IR~~.
\ee
Summarizing the discussion above gives the following:

\begin{prop}
Consider the cosmological model parameterized by
$\fM=(M_0,\cM,\cG,\Phi)$ and define its {\em classical effective
potential} by $V\eqdef M_0\sqrt{2\Phi}$. For each $u\in T\cM$ with
$\pi(u)=m\in \cM$, the maximal cosmological curve $\varphi_u$ of the
model which satisfies:
\be
\varphi_u(0)=m~~\mathrm{and}~~\dot{\varphi}_u(0)=u
\ee
is approximated to first order of the IR expansion by the gradient
flow curve $\eta_m$ of the {\em effective scalar triple}
$(\cM,\cG,V)$ which satisfies:
\be
\eta_m(0)=m~~.
\ee
This approximation is optimal for small $|t|$ when $u\in \Grad_\cG
V$. Accordingly, the cosmological flow $\Pi:\cD\rightarrow T\cM$ of the
model is approximated to first order of the IR expansion by the map
$\Pi_{\IR}:{\hat \cD}_\IR\rightarrow \Grad_\cG V$ defined through:
\be
\Pi_{\IR}(t,u)\eqdef -(\grad_\cG V)(\Pi_V(t,\pi(u)))~~\forall (t,u)\in \cD_\IR~,
\ee
where $\Pi_{\cG,V}:\cD_V\rightarrow \cM$ is the gradient flow of the
effective scalar triple $(\cM,\cG,V)$, whose maximal domain of
definition we denote by $\cD_{\cG,V}\subset \R\times \cM$. Here:
\be
{\hat\cD}_\IR\eqdef \{(t,u)\in T\cM~\vert~(t,\pi(u))\in \cD_{\cG,V}\}~~.
\ee
\end{prop}

\noindent As mentioned above, the first order IR approximation of a
cosmological curve $\varphi$ is optimal when $|t|\ll 1$ for those
cosmological curves which satisfy $\dot{\varphi(0)}\in \Grad_\cG V$.
This motivates the following:

\begin{definition}
A cosmological curve $\varphi$ of the model parameterized by
$\fM=(M_0,\cM,\cG,\Phi)$ is called {\em infrared optimal} if its orbit
meets the gradient flow shell $\Grad_\cG V$ of the effective scalar
triple $(\cM,\cG,V)$, where $V\eqdef M_0\sqrt{2\Phi}$.
\end{definition}

\noindent Suppose that $\varphi:I\rightarrow \cM$ is an infrared
optimal cosmological curve and let $t_0\in I$ be such that
$\dot{\varphi}(t_0)=-(\grad_\cG V)(\varphi(t_0))$. Shifting $t$ by a
constant we can assume that $t_0=0$ without loss of generality. Then
the first order IR approximant of $\varphi$ satisfies
$\varphi_\IR(0)=\varphi(0)$ and
$\dot{\varphi}_\IR(0)=\dot{\varphi}(0)$. Thus $\varphi_\IR$ osculates
in first order to $\varphi$ at $t=0$ and hence approximates $\varphi$
to first order in $t$ for $|t|\ll 1$. Notice that the covariant
accelerations of $\varphi$ and $\varphi_\IR$ need not agree at $t=0$
and hence the approximation need not hold to second order in $t$.

\begin{remark}
Equation \eqref{Vgf} can also be written as:
\be
\frac{\dd \varphi_{\IR}}{\dd \tau}+(\grad \Phi)(\varphi_{\IR}(\tau))=0~~,
\ee
where $\tau$ is the increasing parameter defined though:
\be
\tau(t)=\tau_0+M_0\int_{t_0}^t \frac{\dd t'}{\sqrt{2\Phi(\varphi_{\IR}(t))}}~~,
\ee
with $\tau_0$ an arbitrary constant. Hence
the cosmological curves of the model parameterized by $\fM$ are
approximated by the gradient flow curves of the original scalar triple
$(\cM,\cG,\Phi)$ up to such a curve-dependent reparameterizations of
the gradient flow curves. It is more natural to work with the effective
potential $V$ since its gradient flow parameter coincides with the
cosmological time $t$.
\end{remark}

\subsection{Consistency conditions for the UV and IR approximations}

\noindent The UV and IR approximations discussed in the
previous subsection are accurate when the conditions given below are
satisfied.

\paragraph{For the UV approximation.}

Comparison of \eqref{eomUV} with the cosmological equation
\eqref{eomsingle} shows that the UV approximation amounts to
neglecting\footnote{This agrees with the fact that the UV expansion is
equivalent with an expansion in the overall scale of $\Phi$.} the
contributions $2\Phi(\varphi(t))$ and $(\grad_{\cG} \Phi)(\varphi(t))$
in \eqref{eomsingle}. Hence this approximation is accurate when the
following conditions are satisfied:
\ben
\label{nuconds}
\nu_{1\varphi}(t)\ll 1~~\mathrm{and}~~\nu_{2\varphi}(t)\ll 1~~,
\een
where we defined the {\em first and second UV parameters} of $\varphi$ through:
\ben
\nu_{1\varphi}(t)\eqdef \frac{2\Phi(\varphi(t))}{||\dot{\varphi}(t)||_\cG^2}~~\mathrm{and}~~\nu_{2\varphi}(t)\eqdef M_0 \frac{||(\dd \Phi)(\varphi(t))||_\cG}{||\dot{\varphi}(t)||^2_\cG}~~.
\een
Indeed, the cosmological equation \eqref{eomsingle} can be written as:
\be
\nabla_t\dot{\varphi}(t)+\frac{1}{M_0} [1+\nu_1(t)]^{1/2}||\dot{\varphi}(t)||_\cG \dot{\varphi}(t)+(\grad_\cG\Phi)(\varphi(t))=0~~.
\ee
The first condition in \eqref{nuconds} allows us to approximate this with:
\ben
\label{eom1}
\nabla_t\dot{\varphi}(t)+\frac{1}{M_0} ||\dot{\varphi}(t)||_\cG \dot{\varphi}(t)+(\grad_\cG\Phi)(\varphi(t))=0~~,
\een
while the second condition allows us to approximate \eqref{eom1} with \eqref{eomUV}. 

When conditions \eqref{nuconds} are satisfied, the cosmological curve
$\varphi$ is well-approximated by the reparameterized geodesic
$\varphi_{\UV}$ of $(\cM,\cG)$ which satisfies \eqref{UVin}. Using
the proper length parameter $s$ on the latter, equation \eqref{ts}
gives:
\ben
\label{tprimes}
t'(s)=\frac{1}{||\dot{\varphi}(0)||_\cG} e^{\frac{s}{M_0}}=\frac{1}{||\dot{\varphi}(0)||_\cG}+\frac{1}{M_0} t(s)~~
\een
and we have $\Bvert\frac{\dd \varphi_{\UV}(s)}{\dd s}\Bvert_\cG=1$. Thus:
\be
||\dot{\varphi}(t)||_\cG\approx ||\dot{\varphi}_{\UV}(t)||_\cG=\frac{1}{t'(s)}=||\dot{\varphi}(0)||_\cG e^{-\frac{s}{M_0}}=\frac{M_0||\dot{\varphi}(0)||_\cG}{M_0+t||\dot{\varphi}(0)||_\cG}
\ee
and:
\ben
\label{nu12}
\nu_{1\varphi}(s)\approx \frac{2\Phi(\varphi(s))}{||\dot{\varphi}(0)||_\cG^2} e^{\frac{2s}{M_0}} ~~\mathrm{and}~~
\nu_{2\varphi}(s)\approx \frac{M_0||(\dd \Phi)(\varphi(s))||_\cG}{||\dot{\varphi}(0)||_\cG^2} e^{\frac{2s}{M_0}}~~,
\een
i.e.:
\beqa
\label{nu12t}
&&\nu_{1\varphi}(t)\approx \frac{2\Phi(\varphi(s))}{M_0^2||\dot{\varphi}(0)||_\cG^2}(M_0+t||\dot{\varphi}(0)||_\cG)^2 ~~\nn\\
&&\nu_{2\varphi}(t)\approx \frac{||(\dd \Phi)(\varphi(s))||_\cG}{||\dot{\varphi}(0)||_\cG^2}(M_0+t||\dot{\varphi}(0)||_\cG)^2~~.
\eeqa
For $t=0$, conditions \eqref{nuconds} require:
\be
\frac{2\Phi(\varphi_{\UV}(0))}{||\dot{\varphi}(0)||_\cG^2}\ll 1~~\mathrm{and}~~M_0\frac{||(\dd \Phi)(\varphi_{\UV}(0))||_\cG}{||\dot{\varphi}(0)||_\cG^2}\ll 1~~,
\ee
while for $t\neq 0$ they constrain the cosmological time interval on
which the approximation is accurate. When these conditions are not
satisfied, the leading UV approximation has to be corrected by higher
order terms in the UV expansion.

\paragraph{For the IR approximation.}

Comparison of \eqref{eomIR} with the cosmological equation
\eqref{eomsingle} shows that the IR approximation amounts to
neglecting the terms $\nabla_t \dot{\varphi}(t)$ and
$||\dot{\varphi}(t)||_\cG^2$ in \eqref{eomsingle}. Hence this
approximation is good when the cosmological curve $\varphi$ satisfies:
\ben
\label{lfcond}
\kappa_{1\varphi}(t)\ll 1~~\mathrm{and}~~\kappa_{2\varphi}(t)\ll 1~~,
\een
where we defined the {\em first and second IR parameters} of $\varphi$
through:
\ben
\label{kappadef}
\kappa_{1\varphi}(t)\eqdef\frac{||\dot{\varphi}(t)||_\cG^2}{2\Phi(\varphi(t))}=\frac{1}{\nu_{1\varphi}(t)}~~\mathrm{and}~~\kappa_{2\varphi}(t)\eqdef
\frac{||\nabla_t \dot{\varphi}(t)||_\cG}{||(\dd
\Phi)(\varphi(t))||_\cG}~~.
\een
The first condition in \eqref{lfcond} requires the kinetic energy of
$\varphi$ to be much smaller than its potential energy. Notice that:
\be
\kappa_{1\varphi}(t)=\kappa_1(\dot{\varphi}(t))~~,
\ee
where the function $\kappa_1:T\cM\rightarrow \R_{>0}$ is defined through:
\be
\kappa_1(u)\eqdef \frac{||u||_\cG^2}{2\Phi(\pi(u)))}~~.
\ee
The cosmological equation takes the form:
\ben
\label{nablaphi}
\nabla_t\dot{\varphi}(t)=-\frac{1}{M_0}\sqrt{2\Phi(\varphi(t))}\left[\sqrt{1+\kappa_{1\varphi}(t)}\dot{\varphi}(t)+(\grad_\cG V)(\varphi(t))\right]~~.
\een
Hence the second IR parameter of $\varphi$ can be written as:
\ben
\label{kappa2varphi}
\kappa_{2\varphi}(t)=
\frac{1}{M_0}\sqrt{2\Phi(\varphi(t))}\frac{||\sqrt{1+\kappa_{1\varphi}(t)}\dot{\varphi}(t)+(\grad_\cG
V)(\varphi(t))||_\cG}{||(\dd
\Phi)(\varphi(t))||_\cG}=\kappa_2(\dot{\varphi}(t))~~,
\een
where the function  $\kappa_2:T\cM\rightarrow \R_{>0}$ is defined through:
\ben
\label{kappa2}
\kappa_2(u)\eqdef  \frac{1}{M_0}\sqrt{2\Phi(\pi(u))}\frac{||\sqrt{1+\kappa_1(u)}u+(\grad_\cG V)(\pi(u))||_\cG}{||(\dd \Phi)(\pi(u))||_\cG}~~.
\een
Let $\eta=\varphi_\IR$ be the first order IR approximant of $\varphi$ at $t=0$. Since $\eta$
satisfies the gradient flow equation of $(\cM,\cG,V)$, we have:
\be
\kappa_{1\eta}(t)\eqdef \kappa_1(\dot{\eta}(t))=M_0^2 \frac{||(\dd \Phi)(\eta(t))||_\cG^2}{4\Phi(\eta(t))^2}={\hat \kappa}_1(\eta(t))~~,
\ee
where the function ${\hat \kappa}_1:\cM\rightarrow \R_{>0}$ is defined through:
\be
{\hat \kappa}_1\eqdef M_0^2 \frac{||\dd \Phi||_\cG^2}{4\Phi^2}~~.
\ee
Since $\grad V=-\frac{M_0}{\sqrt{2\Phi}}\grad\Phi$, we
have\footnote{We temporarily denote $||~||_\cG$ by $||~||$ and
$\grad_\cG$ by $\grad$ to simplify notation.}:
\ben
\label{nablaeta}
\nabla_t\dot{\eta}(t)=\nabla_{\dot{\eta}(t)}\dot{\eta}(t)=(\nabla_{\grad V}\grad V)(\eta(t))=
\left[\frac{M_0^2}{2\Phi}\nabla_{\grad\Phi} \grad\Phi-\kappa_1 \grad \Phi \right]\Big{\vert}_{\eta(t)}
\een
because $\nabla_{\grad\Phi}\Phi=(\dd \Phi)(\grad\Phi)=||\dd \Phi||^2$. Thus:
\be
\kappa_{2\eta}(t)\eqdef \frac{||\nabla_t \dot{\eta}(t)||_\cG}{||(\dd
\Phi)(\eta(t))||_\cG}=\frac{M_0^2}{2\Phi}\!\!\left[\!\frac{||\nabla_{\grad\Phi}\grad\Phi||^2}{||\dd\Phi||^2}\!\!+\!\frac{1}{4\Phi^2}||\dd\Phi||^4\!
-\! \frac{1}{2\Phi} \nabla_{\grad\Phi}||\dd \Phi||^2
\!\right]^{1/2}\!\!\Big{\vert}_{\eta(t)}={\hat \kappa}_2(\eta(t))~~,
\ee
where we used the relation:
\be
\cG(X,\nabla_X X)=\frac{1}{2}\nabla_X ||X||^2~~,
\ee
which holds for any vector field $X\in \cX(\cM)$ since $\cG$ is
covariantly constant. Here the function ${\hat
\kappa}_2:\cM\rightarrow \R_{>0}$ is defined through:
\beqa
{\hat \kappa}_2&\eqdef&\frac{||\frac{M_0^2}{2\Phi}\nabla_{\grad\Phi} \grad\Phi-\kappa_1 \grad \Phi||}{||\dd \Phi||}\nn\\
&=&\frac{M_0^2}{2\Phi}\left[\frac{||\nabla_{\grad\Phi}\grad \Phi||^2}{||\dd\Phi||^2}\!\!+\!\frac{1}{4\Phi^2}||\dd\Phi||^4 -\frac{1}{2\Phi}
\nabla_{\grad\Phi}||\dd \Phi||^2 \!\right]^{1/2}
\eeqa
On the other hand, we have::
\ben
\label{kappa2doteta}
\kappa_2(\dot{\eta}(t))=\sqrt{1+\kappa_{1\eta}(t)}-1~~,
\een
where we used \eqref{kappa2} and the fact that $\eta$ satisfies the
gradient flow equation of $(\cM,\cG,V)$.

Using the first order IR approximation for $\varphi$ and its first
time derivative amounts to replace $\varphi(t)$ by $\eta(t)$ and
$\dot{\varphi}(t)$ by $\dot{\eta}(t)$. Then $\kappa_{1\varphi}(t)$ is
replaced by $\kappa_{1\eta}(t)$ and
$\kappa_{2\varphi}(t)=\kappa_2(\dot{\varphi}(t))$ is replaced by
$\kappa_2(\dot{\eta}(t))$. Accuracy of this
approximation requires $\kappa_{1\eta}(t)\ll 1$, which implies
$\kappa_2(\dot{\eta}(t))\ll 1$ by \eqref{kappa2doteta}. For the
approximation to be accurate up to the {\em second} time derivative of
$\varphi$ and $\eta$, we must also have $\kappa_{2\varphi}(t)\approx
\kappa_{2\eta}(t)$, which requires $\kappa_{2\eta}(t)\ll 1$. 

\begin{remark}
Suppose that $\varphi$ is an infrared optimal cosmological curve. In
this case, we have $\varphi(0)=\eta(0):=m\in \cM$ and
$\dot{\varphi}(0)=\dot{\eta}(0)$, which gives:
\be
\kappa_{1\varphi}(0)=\kappa_{1\eta}(0)=\kappa_1(m)~~\mathrm{and}~~\kappa_{2\varphi}(0)=\kappa_2(\dot{\eta}(0))=\sqrt{1+{\hat \kappa}_1(m)}-1~~.
\ee
The curves $\varphi$ and $\eta$ osculate in order {\em two} at the
point $m$ iff $\nabla_t\dot{\varphi}(0)=\nabla_t\dot{\eta}(0)$,
i.e. (see \eqref{nablaphi} and \eqref{nablaeta}):
\ben
\label{osc2}
(\sqrt{1+{\hat \kappa}_1(m)}+{\hat \kappa}_1(m)-1)(\grad_\cG \Phi)=\frac{M_0^2}{2\Phi(m)} (\nabla_{\grad\Phi} \grad\Phi)(m)~~.
\een
This requires  $\kappa_{2\varphi}(0)=\kappa_{2\eta}(0)$, i.e.:
\be
{\hat \kappa}_2(m)=\sqrt{1+{\hat \kappa}_1(m)}-1~~.
\ee
When \eqref{osc2} is satisfied, $\varphi(t)$ is approximated by
$\eta(t)$ to order {\em two} in $t$ for $|t|\ll 1$. 
\end{remark}

\subsection{Relation to some other approximations used in cosmology}

\noindent Recall that the slow roll parameter of a cosmological
curve $\varphi$ is defined through:
\be
\bvarepsilon_{\varphi}(t)\eqdef -\frac{\dot{H_\varphi}(t)}{H_\varphi(t)^2}=-\frac{1}{H_\varphi(t)}\frac{\dd}{\dd t} \log H_\varphi(t)~~.
\ee
We have:
\be
\frac{\dd}{\dd t} \log H_\varphi(t)=\frac{\cG(\dot{\varphi}(t),\nabla_t \dot{\varphi}(t))+(\dd \Phi)(\varphi(t))(\dot{\varphi}(t))}{||\dot{\varphi}(t)||_\cG^2+2\Phi(\varphi(t))}=
-\frac{3H_\varphi(t)||\dot{\varphi}(t)||_\cG^2}{||\dot{\varphi}(t)||_\cG^2+2\Phi(\varphi(t))}~~,
\ee
where in the last equality we used the cosmological equation and the
relation: 
\be
\cG(\dot{\varphi}(t),(\grad_\cG \Phi)(\varphi(t)))=(\dd
\Phi)(\varphi(t))(\dot{\varphi}(t))
\ee 
to simplify the numerator. Thus:
\be
\bvarepsilon_{\varphi}(t)=\frac{3||\dot{\varphi}(t)||_\cG^2}{||\dot{\varphi}(t)||_\cG^2+2\Phi(\varphi(t))}=\frac{3}{1+\nu_{1\varphi}(t)}~~,
\ee
where
$\nu_{1\varphi}(t)=\frac{2\Phi(\varphi(t))}{||\dot{\varphi}(t)||_\cG^2}=\frac{1}{\kappa_{1\varphi}(t)}$
is the first UV parameter and $\kappa_{1\varphi}(t)$ is the first IR
parameter of $\varphi$. Hence the slow roll condition
$\bvarepsilon_\varphi(t)\ll 1$ holds iff $\nu_{1\varphi}(t)\gg 1$
i.e. $\kappa_{1\varphi}(t)\ll 1$. In particular, the slow roll
approximation does not hold when the UV approximation is accurate but
it automatically applies when the first order IR approximation is
accurate up to first time derivatives. Notice, however, that
conditions \eqref{lfcond} for the IR approximation to be accurate up
to second time derivatives are stronger (i.e. more restrictive) than
the slow roll condition, since they also require the parameter
$\kappa_{2\varphi}(t)$ to be small. In particular, the slow roll
approximation by itself is not equivalent with the first order
IR approximation.

Define the {\em scalar gradient flow parameter} of $\varphi$ through:
\be
\veta_{\varphi}(t)\eqdef \frac{1}{3}||\eta_\varphi(t)||_\cG=\frac{1}{3H_\varphi(t)}\frac{||\nabla_t \dot{\varphi}(t)||_\cG}{||\dot{\varphi}(t)||_\cG}
\ee
where $\eta_\varphi(t)$ is the vector gradient flow parameter of
\cite[Sec. 1.5]{genalpha}. Since
$||\nabla_t\dot{\varphi}(t)||_\cG=\kappa_{2\varphi} (t) ||(\dd
\Phi)(\varphi(t))||_\cG$ and
$||\dot{\varphi}(t)||_\cG=\sqrt{2\Phi(\varphi(t))\kappa_{1\varphi}(t)}$, we
have:
\be
H_\varphi(t)=\frac{1}{3M_0}\sqrt{2\Phi(\varphi(t))[1+\kappa_{1\varphi}(t)]}
\ee
and:
\be
\veta_{\varphi}(t)=\frac{M_0}{2}\frac{||(\dd \Phi)(\varphi(t))||_\cG}{\Phi(\varphi(t))} \frac{\kappa_{2\varphi}(t)}{\sqrt{\kappa_1(t)[1+\kappa_1(t)]}}~~.
\ee
When the IR approximation is accurate to second order in time
derivatives, we have $\kappa_{1\varphi}(t),\kappa_{2\varphi}(t)\ll 1$
and the previous relation gives:
\ben
\label{vetaIR}
\veta_{\varphi}(t)\approx \frac{M_0}{2}\frac{||(\dd \Phi)(\varphi(t))||_\cG}{\Phi(\varphi(t))} \frac{\kappa_2(t)}{\sqrt{\kappa_1(t)}}~~.
\een
Hence the gradient flow approximation of \cite{genalpha} applies
within the IR approximation only when this quantity is small. For
two-field models, it was shown in \cite[Sec. 1.9]{genalpha} that the
SRST approximation of \cite{PT1,PT2} is a further specialization of
the gradient flow approximation which applies only when certain
conditions on the Hessian of $\Phi$ are satisfied; these conditions
are much stronger than those for accuracy of the IR
approximation. Hence the SRST approximation can be applied within the
IR approximation only under very restrictive conditions. Thus the IR
approximation is conceptually and quantitatively different from the
gradient flow approximation of \cite{genalpha}. It is markedly
different from (and much less restrictive than) the SRST approximation
and its $n$-field variant.

\section{Renormalization group and dynamical universality classes}
\label{sec:ren}

In this section, we discuss the action of the subgroup $T_\ren\eqdef
\Stab_T(\Phi)$ of the universal similarity group $T$ of Subsection
\ref{subsec:contsym} on the parameters $(M_0,\cG)$ of a multifield
cosmological model with target manifold $\cM$, showing that it plays a
role similar to that of the renormalization group action in the theory
of critical phenomena. This action defines the {\em RG flow} of the
model, which proceeds by homothety transformations of the rescaled
Planck mass $M_0$ and of the scalar manifold metric $\cG$. We also
show that the first order infrared approximation is invariant up to
reparameterization under Weyl transformations of $\cG$ and define UV
and IR universality classes of classical cosmological models. After
briefly discussing the infrared phase structure which arises from the
classification of limit points of cosmological curves, we use these
results and the uniformization theorem of Poincar\'e to show that
two-field models whose scalar manifold metric has constant curvature
are IR universal in the sense that they provide distinguished
representatives of the universality classes of all two-field models.

\subsection{RG similarities and RG transformations}

\noindent For fixed target space $\cM$, the cosmological equation
\eqref{eomsingle} is invariant under the following action of the
multiplicative group $\R_{>0}$ on the curve $\varphi$ and on the
parameters $(M_0,\cG,\Phi)$:
\ben
\label{RGsym}
\varphi\rightarrow \varphi_\epsilon~~,~~M_0\rightarrow \epsilon M_0~~,~~\cG\rightarrow \epsilon^2 \cG~~,~~\Phi\rightarrow \Phi~~(\epsilon >0)~~.
\een
The {\em RG similarity} \eqref{RGsym} is the composition of the
parameter homothety \eqref{parhom} at parameter $\lambda=\epsilon^2$
with the scale similarity \eqref{sim} at parameter $\epsilon$. These
transformations give the restriction $\rho_\ren$ of the
similarity action $\rho$ of the group $T=(\R_{>0})^2$ to the
renormalization subgroup $T_\ren=\Stab_T(\Phi)\simeq \R_{>0}$
discussed in Subsection \ref{subsec:contsym}:
\be
\rho_\ren(\epsilon)(\varphi,M_0,\cG,\Phi)\eqdef \rho(\epsilon^2,\epsilon)(\varphi,M_0,\cG,\Phi)=(\varphi_\epsilon, \epsilon M_0,\epsilon^2\cG,\Phi)~~(\epsilon>0)~~.
\ee

The RG similarities \eqref{RGsym} can be used to absorb $M_0$ into the
overall scale of $\cG$ or viceversa without changing the scale of the
scalar potential $\Phi$. Notice that RG similarities are adiabatic
in the sense that they preserve the total energy
$E_\varphi^{\cG,\Phi}(t)\eqdef
\frac{1}{2}||\dot{\varphi}(t)||_\cG^2+\Phi(\varphi(t))$ of a
cosmological curve $\varphi$. Namely, we have:
\be
E^{\epsilon^2\cG,\Phi}_{\varphi_\epsilon}(t)=E_\varphi^{\cG,\Phi}(t/\epsilon)~~.
\ee

The time $t$ cosmological curves of the model with parameters
$(\epsilon M_0,\epsilon^2 \cG,\Phi)$ coincide with the time
$t/\epsilon$ cosmological curves of the model with parameters
$(M_0,\cG,\Phi)$. Hence the infrared and ultraviolet limits of the
latter can be described equivalently by taking $\epsilon$ to zero and
infinity in the former. Due to these properties, the {\em RG
transformations}:
\ben
\label{RG}
M_0\rightarrow \epsilon M_0~~,~~\cG\rightarrow \epsilon^2 \cG~~,~~\Phi\rightarrow \Phi
\een
play a role akin to that familiar from the theory of critical
phenomena. These transformations give the restriction $\rho_\RG$ of
the parameter action $\rho_{\param}$ of $T$ (see \eqref{rhoparam}) to
the renormalization group $T_\ren$:
\be
\rho_\RG(\epsilon)(M_0,\cG,\Phi)=\rho_\param(\epsilon^2,\epsilon)(M_0,\cG,\Phi)=(\epsilon M_0,\epsilon^2\cG,\Phi)~~(\epsilon>0)~~.
\ee
Under RG transformations, the rescaled scalar field metric
$\cG_0=\frac{1}{M_0^2}\cG$ is invariant, while the classical effective
potential $V=M_0\sqrt{2\Phi}$ and its gradient with respect to $\cG$ change as:
\be
V\rightarrow \epsilon V~~,~~\grad_\cG V\rightarrow \frac{1}{\epsilon}\grad_\cG V~~.
\ee
Hence $\grad_\cG V$ tends to zero in the UV limit $\epsilon\rightarrow
\infty$, while it tends to a current supported on the non-critical
locus $\cM_0=\cM\setminus \Crit\Phi$ in the strict IR limit
$\epsilon\rightarrow 0$. Notice that the geodesic flow of
$(\cM,\cG_0)$ is invariant under RG transformations, since the
covariant derivative of $\cG_0$ satisfies:
\be
\nabla^{\epsilon^2\cG_0}=\nabla^{\cG_0}~~.
\ee
This also follows from the fact that the cosmological equation is
invariant under RG similarities while the geodesic equation is
invariant under affine reparameterizations. On the other hand, the
gradient flow of $(\cM,\cG,V)$ is invariant under RG transformations
up to a reparameterization $\varphi_{\IR}\rightarrow
\varphi_{\IR,\epsilon}$ with positive constant factor $1/\epsilon$;
this also follows from invariance of the cosmological equation under
RG similarities.

\subsection{The dynamical renormalization group flow}

\noindent Let $\cM$ be a smooth manifold and $\cT_2(\cM)\eqdef
\Gamma(\Sym^2(T^\ast \cM))$ be the infinite-dimensional space of
smooth symmetric covariant 2-tensor fields defined on $\cM$. A {\em
Riemannian homothety line} on $\cM$ is a one-dimensional linear
subspace $L\subset \cT_2(\cM)$ which contains a Riemannian metric
defined on $\cM$. In this case, all elements of $L$ which are
positively-homothetic with $\cG$ are Riemannian metrics on $\cM$; such
elements form an open half-line $L_+$ contained in $L$ which satisfies
$L=L_+\cup (-L_+)\cup\{0\}$.  The {\em cosmological homothety plane}
defined by $L$ is the linear space $\Pi(\cM,L)\eqdef \R\oplus L\subset
\R\times \cT_2(\cM)$, which contains the {\em cosmological homothety
cone} $\rC(\cM,L)\eqdef \R_{>0}\oplus L_+$. The {\em cosmological RG
action} on $\rC(\cM,L)$ is the action of the group $T_\ren=\R_{>0}$
defined through:
\be
\rho_\RG(\epsilon)(M_0,\cG)\eqdef (\epsilon M_0,\epsilon^2\cG)~~\forall (M_0,\cG)\in \rC(\cM,L)~~\forall \epsilon>0~~.
\ee
Setting $\epsilon=e^\lambda$ with $\lambda\in \R$, this action
describes the flow on the homothety cone of the Euler vector field
$E_L$ defined through:
\be
E_L(M_0,\cG)\eqdef M_0\oplus 2\cG\in \Pi(\cM,\cG)\equiv T_{M_0,\cG} \rC(\cM,L)~~\forall (M_0,\cG)\in \rC(\cM,L)~~.
\ee
This flow is called the {\em cosmological RG flow} of $(\cM,L)$.

Any choice of a reference metric $\cG_{\mathrm{ref}}\in L_+$ induces
coordinates $w_1,w_2$ on $\rC(\cM,L)$ given by:
\be
M_0=w_1~~,~~\cG=w_2\cG_{\mathrm{ref}}~~,
\ee
which extend to coordinates on the homothety plane $\Pi(\cM,L)$ and
hence identify the latter with $\R^2$ and $\rC(\cM,L)$ with the first
quadrant. Then $\rho_\RG$ identifies with the action:
\be
\rho_\RG(\epsilon)(w_1,w_2)=(\epsilon w_1, \epsilon^2 w_2)~~\forall \epsilon>0
\ee
and $E_L$ identifies with the vector field $E$ on $\R_{>0}^2$ given
by:
\be
E(w_1,w_2)=(w_1,2w_2)~~\forall w_1,w_2>0. 
\ee
Moreover, the integral curves of the RG flow identify with the
solutions of the system:
\beqan
\label{RGflow}
\frac{\dd w_1(\lambda)}{\dd \lambda}&=&E_1(w_1(\lambda),w_2(\lambda))=w_1(\lambda)\nn\\
\frac{\dd w_2(\lambda)}{\dd \lambda}&=&E_2(w_1(\lambda),w_2(\lambda))=2 w_2(\lambda)~~,
\eeqan
namely:
\be
w_1(\lambda)=e^\lambda w_1(0)~~,~~w_2(\lambda)=e^{2\lambda} w_2(0)~~.
\ee
The limit $\lambda\rightarrow +\infty$ recovers the UV limit
$\epsilon\rightarrow +\infty$ while $\lambda\rightarrow -\infty$
corresponds to the IR limit $\epsilon\rightarrow 0$.  These limits
correspond to the fixed points of the RG flow on the one-point
compactification of the closure of the homothety cone, which are the
apex of the cone and the point at infinity (see Figure
\ref{fig:RGflow}).

\begin{figure}[H]
\centering
\begin{minipage}{.45\textwidth}
\centering ~~~~ \includegraphics[width=.9\linewidth]{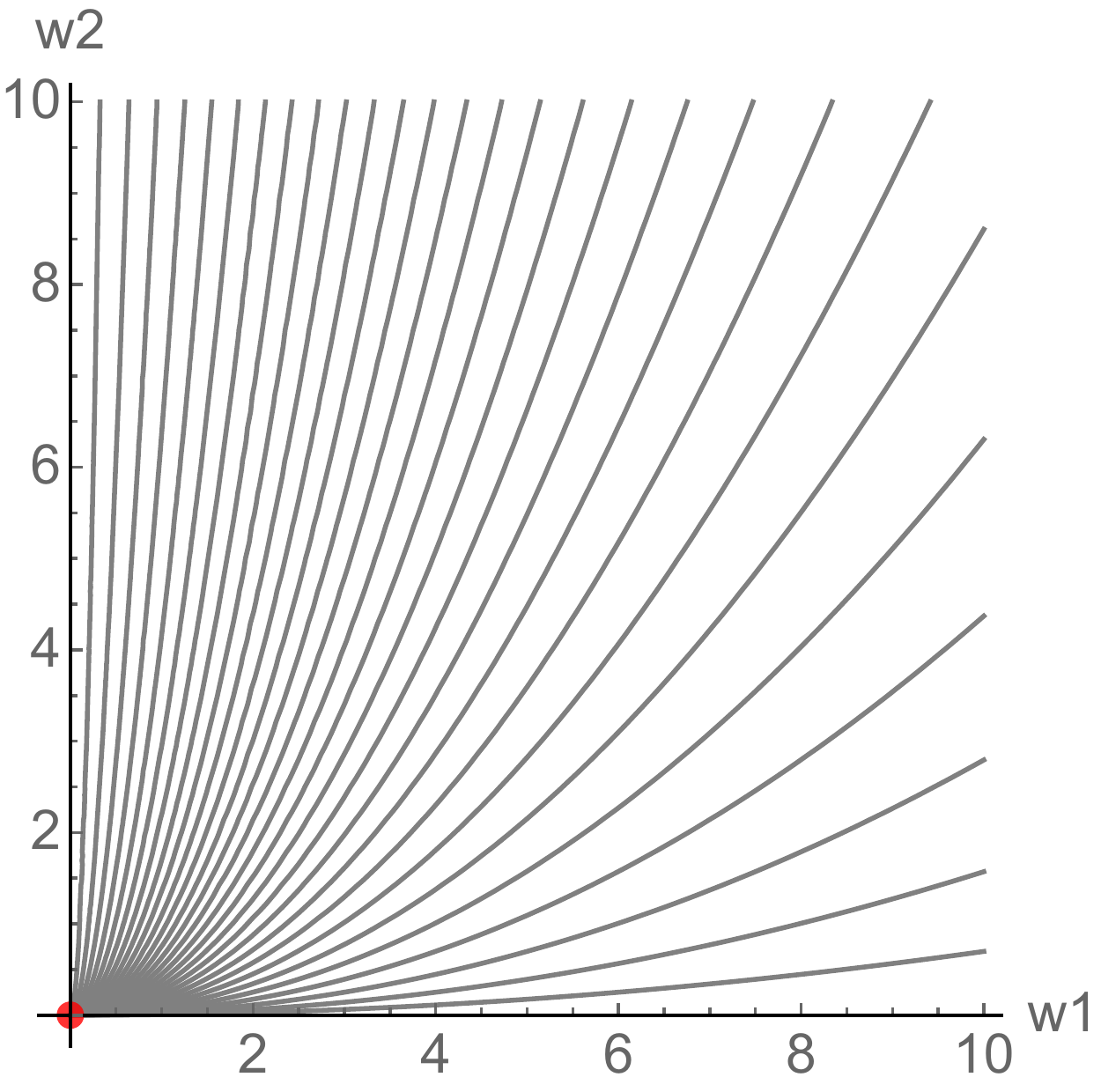}
\subcaption{Integral curves of the RG flow on the homothety cone.}
\end{minipage}
\hfill
\begin{minipage}{.5\textwidth}
\vskip 0.6 em
\centering ~~~~ \includegraphics[width=.8\linewidth]{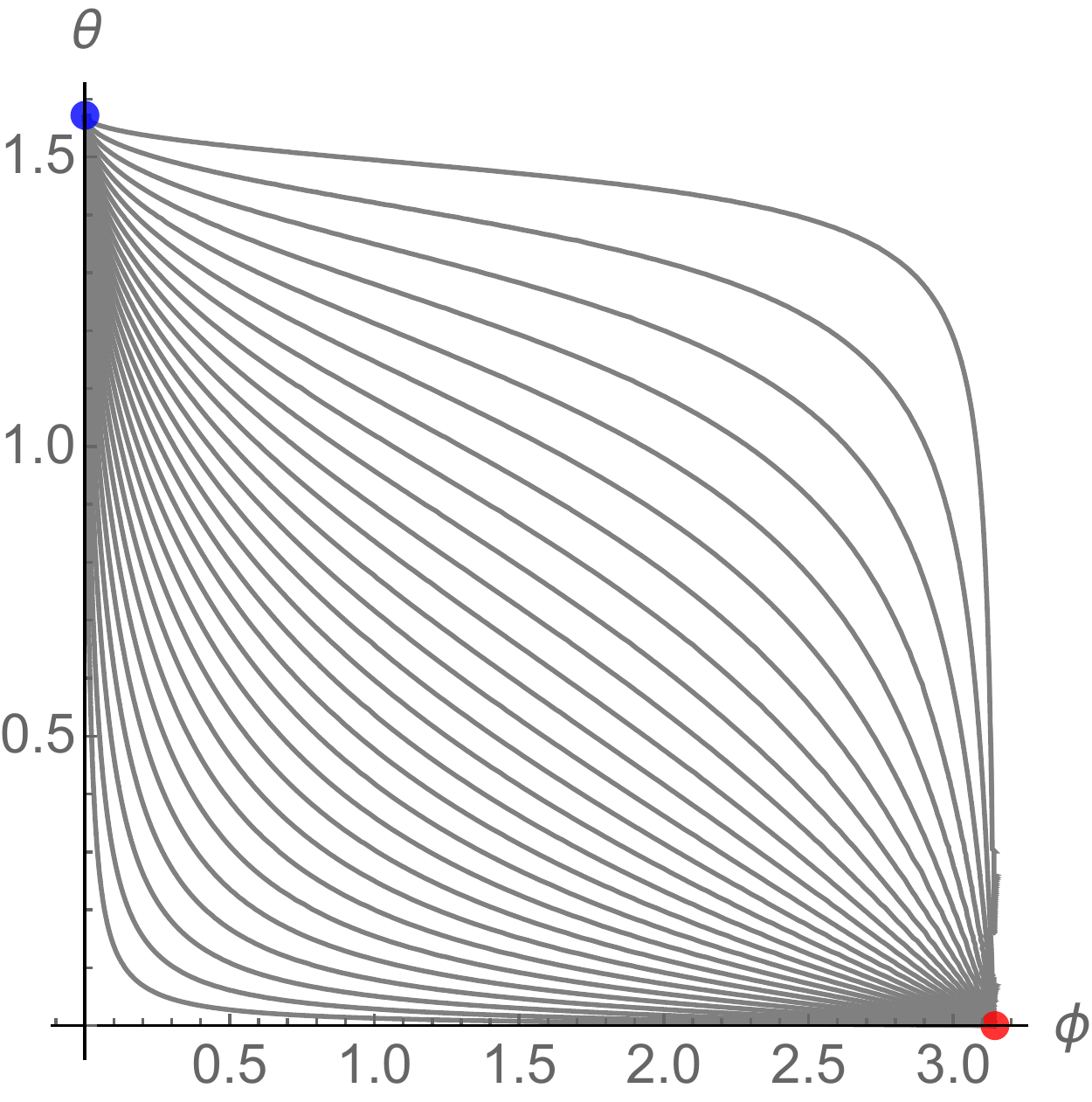}
\subcaption{Integral curves of the RG flow on the one-point
compactification of the homothety cone.}
\end{minipage}
\caption{Integral curves of the RG flow on the homothety cone and on
the one-point compactification of its closure. In the second figure we
identified the Alexandroff compactification of the homothety plane
with the unit sphere through stereographic projection. Here $\phi\in
[0,\pi]$ is the spherical altitude angle on the unit sphere and
$\theta\in [0,2\pi]$ is the azimuth/longitude angle (which coincides
with the polar angle in the homothety plane), while
$r=\sqrt{w_1^2+w_2^2}=\cot(\frac{\phi}{2})$ is the distance from the
origin in the homothety plane. The homothety cone identifies with the
region $(\phi,\theta)\in [0,\pi]\times [0,\frac{\pi}{2}]$ of the unit
sphere. The RG flow on the one-point compactification has fixed points
at $(\phi,\theta)=(\pi,0)$ (the red dot) and
$(\phi,\theta)=(0,\frac{\pi}{2})$ (the blue dot), which correspond
respectively to the apex of the homothety cone and its point at
infinity, i.e. to the south and north poles of the sphere. These fixed
points give the IR limit (the red dot) and the UV limit (the blue
dot).}
\label{fig:RGflow}
\end{figure}

\noindent Consider a model parameterized by $(M_0,\cM,\cG,\Phi)$,
where $\cG\in L_+$.  The discussion of the previous subsection implies
that the cosmological RG flow curve of $(\cM,L)$ which passes through
the point $(M_0,\cG)\in \rC(\cM,L)$ induces a curve in the space of
all flows defined on $T\cM$ which interpolates when $\lambda$ runs
from $-\infty$ to $+\infty$ between the gradient flow of $(\cM,\cG,V)$
(where $V=M_0\sqrt{2\Phi}$) and a modification of the geodesic flow of
$(\cM,\cG_0)$ (where $\cG_0=\frac{1}{M_0^2}\cG$).

\subsection{Conformal invariance in the infrared}
\label{subsec:Conformal}

\noindent The gradient flow of a scalar potential $V$ defined on a
Riemannian manifold $(\cM,\cG)$ is invariant under Weyl rescalings of
the metric $\cG$ up to reparameterization of the gradient flow
curves. More precisely, consider a solution $\eta:I\rightarrow \cM$ of
the gradient flow equation:
\ben
\label{GradientFlow}
\frac{\dd \eta(t)}{\dd t}=-(\grad_{\cG} V)(\eta(t))
\een
of $V$ with respect to the metric $\cG$ and let $\Omega$ be a smooth
and everywhere-positive real-valued function defined on $\cM$. Then
the $\eta$-dependent increasing reparameterization $t\rightarrow
\tau$ defined through:
\be
\tau(t)\eqdef \tau_\eta(t)=\int_{t_0}^t \dd t' \Omega(\eta(t'))+C~~,
\ee
(where $t_0\in I$ and $C$ is an arbitrary constant) takes $\eta(t)$
into a solution $\eta(\tau)$ of the gradient flow equation:
\be
\frac{\dd \eta(\tau)}{\dd \tau}=-(\grad_{\cG'} V)(\eta(\tau))
\ee
of $V$ with respect to the metric $\cG'\eqdef \Omega \cG$. This
implies that the gradient flow orbits of $V$ with respect to the
metrics $\cG$ and $\Omega \cG$ coincide as oriented submanifolds of
$\cM$. In particular, these orbits depend only on the Weyl-equivalence
class of $\cG$.

This observation can be generalized as follows. Recall that a smooth
map $f:(\cM_1,\cG_1)\rightarrow (\cM_2,\cG_2)$ between Riemannian
manifolds $(\cM_1,\cG_1)$ and $(\cM_2,\cG_2)$ is called a conformal
diffeomorphism if $f^\ast(\cG_2)=\Omega \cG_1$ for some positive
smooth function $\Omega:\cM_1\rightarrow \R_{>0}$.

\begin{definition}
Two scalar triples $(\cM_1,\cG_1,V_1)$ and $(\cM_2,\cG_2,V_2)$ are
called smoothly {\em conformally equivalent} if there exists a
conformal diffeomorphism $f:(\cM_1,\cG_1)\rightarrow (\cM_2,\cG_2)$
such that $V_1=V_2\circ f$. In this case, $f$ is called a (smooth)
{\em conformal equivalence} between the two triples. A {\em conformal
automorphism} of a scalar triple $(\cM,\cG,V)$ is a conformal
equivalence from $(\cM,\cG,V)$ to itself.
\end{definition}

\noindent With this definition, we show in Appendix \ref{app:grad}
that the gradient flows of conformally-equivalent scalar triples are
smoothly topologically equivalent. In particular, the smooth
topological equivalence class of the gradient flow of a scalar triple
is invariant under conformal automorphisms of that triple. Since in
the gradient flow of the function $V=M_0\sqrt{2\Phi}$ with respect to
$\cG$ gives the first order IR approximant of cosmological curves for
the model parameterized by $(M_0,\cM,\cG,\Phi)$, we conclude that:

\

\noindent {\em Up to topological equivalence, the first order IR
approximant of the cosmological flow of a multifield cosmological
model parameterized by $(M_0,\cM,\cG,\Phi)$ depends only on the
conformal equivalence class of the effective scalar triple
$(\cM,\cG,V)$, where $V\eqdef M_0\sqrt{2\Phi}$. Up to curve-dependent
increasing reparameterization of the flow curves, this approximant
depends only on $\cM$, $V$ and on the Weyl-equivalence class of
$\cG$.}

\subsection{Dynamical universality classes}

\noindent Consider two cosmological models parameterized by
$\fM_i=(M_{0i},\!\cM_i,\!\cG_i,\!\Phi_i)$ ($i=1,2$) whose tangent
bundle projections and Levi-Civita connections we denote by $\pi_i$
and $\nabla_i$. Let $\cG_{0i}\eqdef \frac{1}{M_{0i}^2}\cG_i$ be the
rescaled scalar manifold metrics and $V_i\eqdef M_{0i}\sqrt{2\Phi_i}$
be the effective scalar potentials of the two models. Notice that $\nabla_i$
coincides with the Levi-Civita connection of $\cG_{0i}$. 

\paragraph{UV conjugations and UV equivalences.}

\noindent Recall the map \eqref{hatf} induced between the sets of
smooth curves in two manifolds by a smooth semispray map.

\begin{definition} A {\em smooth geodesic conjugation} between two
connected Riemannian manifolds $(\cM_1,\cG_{01})$ and
$(\cM_2,\cG_{02})$ is a smooth semispray map $f:T\cM_1\rightarrow
T\cM_2$ which restricts to a topological conjugation between their
normalized geodesic flows, i.e. $f$ restricts to a diffeomorphism
$f_S:ST\cM_1\rightarrow ST\cM_2$ between the unit sphere bundles
$ST\cM_1$ and $ST\cM_2$ and has the property that ${\hat f}(\psi_1)$
is a normalized geodesic of $(\cM_2,\cG_{02})$ whenever $\psi_1$ is a
normalized geodesic of $(\cM_1,\cG_{01})$.
\end{definition}

\begin{eg}
Consider a diffeomorphism $f_0:\cM_1\rightarrow \cM_2$. Then it is
easy to see that $f_0$ maps normalized geodesics of $(\cM_1,\cG_{01})$
into normalized geodesics of $(\cM_2,\cG_{02})$ iff it is an isometry,
which in turn is equivalent with the condition that its differential
$\dd f_0$ maps $ST\cM_1$ into $ST\cM_2$. In this case, the restriction
$\dd f_0\vert_{ST\cM_1}:ST\cM_1\rightarrow ST\cM_2$ is a geodesic
conjugation.
\end{eg}

\begin{remark} Recall that an {\em affine mapping} from
$(\cM_1,\cG_{01})$ to $(\cM_2,\cG_{02})$ is a diffeomorphism
$f_0:\cM_1\rightarrow \cM_2$ such that $f_0^\ast(\nabla_2)=\nabla_1$.
Such a map takes affinely parameterized geodesics of
$(\cM_1,\cG_{01})$ into affinely parameterized geodesics of
$(\cM_2,\cG_{02})$ but it need not take normalized geodesics into
normalized geodesics. An affine mapping from a Riemannian manifold
$(\cM,\cG_0)$ to itself is called an {\em affine transformation} of
$(\cM_0,\cG_0)$; such transformations form a Lie group. An
infinitesimal affine mapping is sometimes called an {\em affine
collineation}.  When $(\cM,\cG_0)$ is complete, irreducible and of
dimension greater than one, any affine transformation of $(\cM,\cG_0)$
is an isometry (see \cite{Kobayashi, IO}).
\end{remark}

\begin{definition} A {\em smooth geodesic equivalence} between two
Riemannian manifolds $(\cM_1,\cG_{01})$ and $(\cM_2,\cG_{02})$ is a
smooth semispray map $f:T\cM_1\rightarrow T\cM_2$ which identifies
geodesic orbits, i.e. $\psi_1:I_1\rightarrow \cM_1$ is an arbitrarily
parameterized geodesic of $(\cM_1,\cG_{01})$ iff ${\hat f}(\psi_1)$ is
an arbitrary parameterized geodesic of $(\cM_2,\cG_{02})$. 
\end{definition}

\begin{eg}
Consider a smooth {\em geodesic (a.k.a. projective) mapping}
$f_0:(\cM_1,\cG_{01})\rightarrow (\cM_2,\cG_{02})$, i.e.  a smooth
diffeomorphism from $\cM_1$ to $\cM_2$ which maps the geodesic orbits
of $(\cM_1,\cG_{01})$ into those of $(\cM_2,\cG_{02})$, a condition
which amounts to the requirement that the connections $\nabla_1$ and
$f_0^\ast(\nabla_2)$ be projectively-equivalent. Then $\dd
f_0:T\cM_1\rightarrow T\cM_2$ is a smooth geodesic equivalence. The
study of geodesic mappings is a classical subject in Riemannian
geometry (see, for example \cite{Mikes}) which is intimately connected
(see \cite{CrampinSaunders}) to Cartan's projective differential
geometry \cite{CartanProj}.
\end{eg}

\noindent The results of Section \ref{sec:ScalingLimits} motivate the
following:

\begin{definition} Consider two classical cosmological models
parameterized by $\fM_1=(M_{01},\!\cM_1,\!\cG_1,\!\Phi_1)$ and
$\fM_2=(M_{02},\!\cM_2,\!\cG_2,\!\Phi_2)$ and let $V_i\eqdef
M_{0i}\sqrt{2\Phi_i}$ and $\cG_{0i}\eqdef \frac{1}{M_{0i}^2}\cG_i$
($i=1,2$). We say that the models are:
\begin{itemize}
\item {\em smoothly UV conjugate} and write $\fM_1\!\equiv_{\UV}\fM_2$ if there
  exists a smooth geodesic conjugation:
\be
f:(\cM,\cG_{01})\rightarrow (\cM,\cG_{02})~~.
\ee
\item {\em smoothly UV equivalent} and write $\fM_1\!\sim_{\UV}\fM_2$ if there
  exists a smooth geodesic equivalence:
\be
f:(\cM,\cG_{01})\rightarrow (\cM,\cG_{02})~~.
\ee
\end{itemize}
In the situations above, $f$ is called respectively a smooth {\em UV
conjugation} or {\em UV equivalence} between $\fM_1$ and $\fM_2$. 
\end{definition}

\paragraph{IR conjugations and IR equivalences.}

The results of Section \ref{sec:ScalingLimits} show
that the IR limits of the two models can be identified when the
gradient flows of the scalar triples $(\cM_1,\cG_1,V_1)$ and
$(\cM_2,\cG_2,V_2)$ are smoothly topologically conjugate, i.e. related
by a diffeomorphism $f:\cM_1\rightarrow \cM_2$ (see Appendix
\ref{app:tsequiv}). This amounts to the condition that the vector fields
$\grad_{\cG_1}V_1$ and $\grad_{\cG_2} V_2$ be $f$-related,
i.e. $\grad_{\cG_2}V_2=f_\sharp (\grad_{\cG_1}V_1)$, where
$f_\sharp:\cX(\cM_1)\rightarrow \cX(\cM_2)$ denotes the push-forward
of vector fields by $f$:
\be
f_\sharp(X)=(\dd f)\circ X\circ f^{-1}~~\forall X\in
\cX(\cM_1)~~.
\ee
Moreover, the two gradient flows can be identified through $f$ up to
increasing reparameterization (and hence are smoothly topologically
equivalent, see Appendix \ref{app:tsequiv}) iff $\frac{1}{\Omega} \,
\grad_{\cG_1}V_1$ is $f$-related to $\grad_{\cG_2} V_2$ for some
positive smooth function $\Omega$ defined on $\cM_1$ (which induces a
Weyl rescaling $\cG_1\rightarrow \Omega\cG_1$ of $\cG_1$) i.e. iff we
have $f_\sharp(\frac{1}{\Omega} \, \grad_{\cG_1}V_1)=\grad_{\cG_2}
V_2$. These observations motivate the following:

\begin{definition}
Consider two scalar triples $(\cM_1,\cG_1,V_1)$ and
$(\cM_2,\cG_2,V_2)$. A smooth diffeomorphism
$f:(\cM_1,\cG_1,V_1)\rightarrow (\cM_2,\cG_2,V_2)$ is called:
\begin{itemize}
\item {\em smooth gradient conjugation}, if it satisfies the condition:
\ben
\label{gradconj}
f_\sharp(\grad_{\cG_1} V_1)=\grad_{\cG_2} V_2~~.
\een
\item {\em smooth gradient equivalence}, if it satisfies the
condition:
\ben
\label{gradequiv}
f_\sharp(\Omega^{-1}\grad_{\cG_1} V_1)=\grad_{\cG_2} V_2
\een
for some positive smooth function $\Omega:\cM_1\rightarrow \R_{>0}$.
\end{itemize}
The two scalar triples are called smoothly {\em gradient conjugate/equivalent}
if there exists a smooth gradient conjugation/equivalence
$f:(\cM_1,\cG_1,V_1)\rightarrow (\cM_2,\cG_2,V_2)$.
\end{definition}

\noindent Since the differential pull-back
$f^\ast:\cX(\cM_2)\rightarrow \cX(\cM_1)$ is given by
$f^\ast=(f_\sharp)^{-1}$, condition \eqref{gradequiv} is equivalent
with:
\ben
\label{gradequiv2}
\grad_{\cG_1} V_1=\Omega~\grad_{f^\ast(\cG_2)} f^\ast(V_2)~~,
\een
where $f^\ast(V_2)=V_2\circ f$. Gradient conjugation and equivalence
are equivalence relations on scalar triples.

\begin{eg}
Any conformal equivalence of scalar triples
$f:(\cM_1,\cG_1,V_1)\rightarrow (\cM_2,\cG_2,V_2)$ is a gradient
equivalence, as can be seen from its defining relations
$f^\ast(\cG_2)=\Omega \cG_1$ and $f^\ast(V_2)=V_1$. However, condition
\eqref{gradequiv} admits solutions $f$ which need not be conformal
equivalences of triples. In local coordinates $(x^1,\ldots, x^d)$ on
$\cM_1$ and $(y^1,\ldots, y^d)$ on $\cM_2$, this condition takes the
form:
\ben
\label{gradequivlocal}
\frac{1}{\Omega(x)}\cG_1^{jk}(x)\pd_j f^i(x) \pd_k V_1(x)=\cG_2^{ij}(f(x))(\pd_j V_2)(f(x))~~,
\een
which is weaker than the defining conditions of a conformal
transformation from $(\cM_1,\cG_1,V_1)$ to $(\cM_2,\cG_2,V_2)$:
\ben
\label{conflocal}
(\cG_{2})_{ij}(f(x))\pd_k f^i(x)\pd_l f^j(x)=\Omega(x) (\cG_{1})_{kl}(x)~~\mathrm{and}~~V_2(f(x))=V_1(x)~~.
\een
Here $f^i(x)=y^i(f(x))$. If one chooses coordinates such that
$x^i=y^i(f(x))$ then $f$ is locally given by $f^i(x)=x^i$ and
\eqref{gradequivlocal} becomes:
\be
\frac{1}{\Omega(x)}\cG_1^{ij}(x) \pd_j V_1(x)=\cG_2^{ij}(x)\pd_j V_2(x)
\ee
while \eqref{conflocal} reduces to:
\be
\frac{1}{\Omega(x)}\cG_1^{ij}(x)=\cG_2^{ij}(x)~\mathrm{and}~~V_1(x)=V_2(x)~~.
\ee
A gradient conjugation of a scalar triple to itself is called a {\em
gradient symmetry}, while a gradient equivalence to itself is called a
{\em weak gradient symmetry} of that triple. The condition that a
diffeomorphism $f:\cM\rightarrow \cM$ be a weak gradient symmetry of
the scalar triple $(\cM,\cG,V)$ reduces to the requirement that the
gradient vector field $v\eqdef \grad_\cG V$ (which has components
$v^i(x)=\cG^{ij}(x)\pd_j V(x)$) satisfies $f_\ast(\frac{1}{\Omega}
v)=v$, i.e.:
\be
\pd_j f^i(x) v^j (x)=\Omega(x)v^i(f(x))~~.
\ee
\end{eg}

\noindent With these preparations, the remarks made above motivate
the following:

\begin{definition}
Consider two classical cosmological models parameterized by
$\fM_1=(M_{01},\!\cM_1,\!\cG_1,\!\Phi_1)$ and
$\fM_2=(M_{02},\!\cM_2,\!\cG_2,\!\Phi_2)$ and let $V_i\eqdef
M_{0i}\sqrt{2\Phi_i}$ ($i=1,2$). We say that the models are 
\begin{itemize}
\item {\em smoothly IR conjugate} and write $\fM_1\equiv_{\IR} \fM_2$ if there
exists a smooth gradient conjugation:
\be
f:(\cM_1,\cG_1,V_1)\rightarrow (\cM_2,\cG_2,V_2)~~.
\ee
\item {\em smoothly IR equivalent} and write $\fM_1 \sim_{\IR}\fM_2$ if
there exists a smooth gradient equivalence  
\be
f:(\cM_1,\cG_1,V_1)\rightarrow (\cM_2,\cG_2,V_2)~~.
\ee
\end{itemize}
In the situations above, $f$ is called respectively a smooth {\em IR
  conjugation} or {\em IR equivalence} between $\fM_1$ and $\fM_2$.
\end{definition}

A (necessarily strict) IR or UV smooth conjugation between a model
parameterized by $(M_0,\cM,\cG,\Phi)$ and itself is called an IR or UV
{\em symmetry} of that model.

The binary relations introduced above are equivalence relations on the
class of multifield cosmological models; in fact, they are the
isomorphism relations of associated groupoids. The equivalence classes
of $\sim_{\UV}$ and $\sim_{\IR}$ are called smooth UV and IR {\em
cosmological universality classes}.

\subsection{The late time infrared phase structure}
\label{subsec:IRModuliSpace}

\noindent As explained in Section \ref{sec:ScalingLimits}, the
cosmological curves of a model parameterized by $(M_0,\cM,\cG,\Phi)$
are approximated by gradient flow curves of the scalar triple
$(\cM,\cG,V)$ in the IR limit.  With sufficiently strong assumptions
on the behavior of $\cM$ and $V$ near the Freudenthal ends of $\cM$,
the $\omega$- extended limit points of such curves coincide with those
of certain gradient flow curves of $(\cM,\cG,V)$. In good cases, a
cosmological curve $\varphi$ approaches a gradient flow curve
$\varphi_\IR^{(+)}$ in the distant future and hence $\varphi$ can be
replaced in the infrared by $\varphi_\IR^{(+)}$ for late times.

This statement is familiar when $\cM$ is compact and $V$ is a Morse
function. In that case, the gradient flow of $(\cM,\cG,V)$ is complete
and hence its maximal flow curves are defined on the entire real
axis. Moreover, any non-constant maximal gradient flow curve $\eta$
satisfies:
\be
\exists \lim_{t\rightarrow -\infty} \eta(q)=c_i~~\mathrm{and}~~\exists
\lim_{t\rightarrow \infty } \eta(q)=c_f~~, \ee where $c_i$ and $c_f$
are critical points of $V$ whose Morse indices obey
$\ind(c_i)>\ind(c_f)$. In this situation, the moduli space of gradient
flow orbits has components indexed by the pair $(c_i,c_f)$, while the
early and late time behavior of $\eta$ depend respectively on $c_i$
and $c_f$. In this case, the late time infrared dynamics of the
cosmological model parameterized by $(M_0,\cM,\cG,\Phi)$ has a {\em
phase structure} in which different phases are indexed by the critical
points of $V$; every cosmological curve ``settles'' at late times in a
phase indexed by a critical point $c$ of Morse index greater that
$d\eqdef \dim\cM$. The asymptotic behavior of the cosmological flow
for late times in the phase indexed by $c$ depends on the Morse index
of $c$ and on the behavior of $\cG$ near $c$.

This simple characterization of limit points of the gradient flow of
$(\cM,\cG,V)$ no longer holds when $V$ is Morse but $\cM$ is
non-compact. With strong enough conditions on the topology of $\cM$
and on the asymptotic behavior of $V$, the gradient flow curves will
still have $\alpha$- and $\omega$- extended limits but each of these
can be either a critical point of $V$ or a Freudenthal end of
$\cM$. In this case, the IR phases of the model are indexed by
critical points of $V$ and by those ends of $\cM$ which can arise as
$\omega$-limit points of a cosmological curve.

For general effective potentials $V$ with non-compact $\cM$, a maximal
cosmological curve can have more than one extended $\omega$-limit
point. In this situation, IR phases can be classified partially by the
allowed $\omega$-limit sets, but specifying these need not provide a
full classification and more detailed analysis is required. In
particular, the classification of such phases can be quite involved
and a simple relation to the late time behavior of the gradient flow
of $(\cM,\cG,V)$ may not exist.

\section{IR universality of two-field models with hyperbolic scalar manifold}
\label{sec:2field}

\noindent The IR conformal invariance property discussed in Subsection
\ref{subsec:Conformal} has a striking implication for two-field
cosmological models when combined with the uniformization theorem of
Poincar\'e. According to the latter, the Weyl equivalence class of any
Riemannian metric $\cG$ defined on a borderless connected surface
$\Sigma$ contains a unique complete metric $G$ (called the {\em
uniformizing metric} of $\cG$) of constant Gaussian curvature $K$
equal to one of the values $-1$, $0$ or $+1$. The case $K=-1$ is
generic; in this case, the metric $\cG$ (and its conformal class) is
called {\em hyperbolizable} and $G$ is called the {\em
hyperbolization} of $\cG$. The cases $K=+1$ and $K=0$ occur only for
seven topologies, as follows:
\begin{itemize}
\item When $K=+1$, the surface $\Sigma$ must be diffeomorphic with the
$2$-sphere $\rS^2$ or with the real projective plane $\R\P^2\simeq
\rS^2/\Z_2$. Both of these surfaces admit a unique metric of unit
Gaussian curvature.
\item When $K=0$, the surface $\Sigma$ must be diffeomorphic with the
2-torus $\rT^2$, the Klein bottle $\rK^2=\R\P^2\times \R\P^2\simeq
\rT^2/\Z_2$, the open annulus $\rA^2$ (which is diffeomorphic with the
open cylinder and with the twice punctured sphere), the open
M\"{o}bius strip $\rM^2\simeq \rA^2/\Z_2$ (which is diffeomorphic with
the once-punctured real projective plane) or with the plane
$\R^2$. The 2-torus admits a three-parameter family of flat metrics
while the Klein bottle admits a two-parameter family of such. The open
annulus and open M\"{o}bius strip admit a one-parameter family of
complete flat metrics while $\R^2$ admits a unique complete flat
metric.
\end{itemize}

\noindent The plane $\R^2$, the open annulus $\rA^2$
and the open M\"{o}bius strip $\rM^2$ are the only three surfaces which
admit {\em two} types of uniformizing metrics, namely a Riemannian
metric defined on one of these surfaces is uniformized either by a
complete flat metric or by a hyperbolic metric depending on its
conformal class. Namely:
\begin{itemize}
\item The plane $\R^2$ admits both a complete flat metric and a
hyperbolic metric (the Poincar\'e metric).
\item The open annulus $\rA^2$ admits a one-parameter family of
complete flat metrics, a hyperbolic metric which makes it into the
hyperbolic punctured disk and a one-parameter family of hyperbolic
metrics which produce the hyperbolic annuli.
\item The open M\"{o}bius strip $\rM^2$ admits a one-parameter family
of complete flat metrics and a one-parameter family of hyperbolic
metrics which is obtained by quotienting the hyperbolic annuli through
a $\Z_2$ group of isometries.
\end{itemize}

\noindent Surfaces on which any metric is hyperbolizable are called of
{\em general type}; these include all borderless connected surfaces
except $\rS^2,\R\P^2, \rT^2,\rK^2,\rA^2,\rM^2$ and $\R^2$. When
$\Sigma$ is diffeomorphic with $\rA^2,\rM^2$ or $\rK^2$, only certain
conformal classes are hyperbolizable, while the remaining conformal
classes uniformize to a complete flat metric. When $\Sigma$ is
diffeomorphic with $\rT^2$ or $\rK^2$, every conformal class
uniformizes to a complete flat metric, while when it is diffeomorphic
with $\rS^2$ or $\R\P^2$ any conformal class uniformizes to a complete
metric with Gaussian curvature equal to $+1$.

Consider a two-field cosmological model parameterized by
$(M_0,\Sigma,\cG,\Phi)$.  The observations of Subsection
\ref{subsec:Conformal} imply that the gradient flow of
$V=M_0\sqrt{2\Phi}$ computed with respect to $\cG$ has the same
oriented orbits as the gradient flow of $V$ computed with respect to
the uniformizing metric $G$ and hence differs from the latter only by
an increasing reparameterization of the gradient flow curves. Thus
models whose scalar field metric has constant Gaussian curvature equal
to $-1$, $0$ or $+1$ provide distinguished representatives of the IR
universality classes of all two-field cosmological models. More
precisely:

\

\noindent {\em Up to (curve-dependent) increasing reparameterization,
the first order IR approximant of the cosmological flow of a two-field
model with scalar triple $(\Sigma,\cG,\Phi)$ and rescaled Planck mass
$M_0$ is described by the gradient flow of the scalar triple
$(\Sigma,G,V)$, where $G$ is the uniformizing metric of $\cG$ and
$V\eqdef M_0\sqrt{2\Phi}$ is the classical effective scalar potential
of the model.}

\

\noindent When $\Sigma$ is of general type, this statement implies
that IR universality classes of two-field models with target $\Sigma$
are classified by gradient equivalence classes of hyperbolic scalar
triples $(\Sigma,G,V)$; more generally, this holds for two-field
models with hyperbolizable scalar manifold $(\Sigma,\cG)$. Two-field
models whose complete scalar manifold metric $\cG$ has constant
negative curvature are called {\em generalized two-field
$\alpha$-attractor models}. In this case, one has $\cG=3\alpha G$ for
some positive constant $\alpha$ and some hyperbolic metric $G$. Such
models were introduced and studied in \cite{genalpha, elem, modular}
and form a very wide extension of the two-field cosmological
$\alpha$-attractor models previously introduced in \cite{KLR}, which
correspond to the topologically trivial case when $(\Sigma,G)$ is the
Poincar\'e disk. It follows that generalized two-field
$\alpha$-attractor models are {\em IR universal} among two-field
cosmological models with hyperbolizable target manifold. This gives a
conceptual reason to single out generalized two-field
$\alpha$-attractor models for special study.

\begin{remark} The notion of IR universality employed in this paper
differs from the colloquial use of the same word in the
$\alpha$-attractor literature (see for example \cite{AKLWY}), where
certain predictions extracted from models whose scalar manifold is a
Poincar\'e disk were claimed to be `universal' in the sense that they
are insensitive to small enough perturbations of the scalar potential
(a notion which corresponds to structural stability of the
cosmological dynamical system and differs conceptually from our notion
of IR universality). Poincar\'e disk models {\em cannot} be IR
universal even in the class of models with contractible target, since
the open disk $\rD^2\simeq \R^2$ admits both metrics which uniformize
to the Poincar\'e metric and metrics which uniformize to the complete
flat metric\footnote{While the Poincar\'e metric is conformally flat,
it is not conformally equivalent to the {\em complete} flat metric on
$\rD^2\simeq \R^2$.}. As we show in a separate publication, IR
universality classes of two-field models with hyperbolizable scalar
manifold are much more complicated than one might expect from such
works. When $\Sigma$ is not topologically trivial, such models can
have distinct late time IR phases associated to the Freudenthal ends
of $\Sigma$, all of which differ from the single IR phase associated
to the plane end of the Poincar\'e disk; in fact, the latter is the
only hyperbolic surface which admits a plane end and hence Poincar\'e
disk models are uniquely {\em special} among two-field models with
hyperbolizable target manifold. Notice that Freudenthal end phases
arise in addition to those indexed by the critical points of the
scalar potential.
\end{remark}

\section{Conclusions and further directions}
\label{sec:Conclusions}

We studied the scaling behavior and scaling limits of classical
multifield cosmological models with scalar triple $(\cM,\cG,\Phi)$ and
rescaled Planck mass $M_0$ and their dynamical renormalization group,
showing that the latter deforms the cosmological flow of such models
into a family of flows defined on $T\cM$ which interpolates between a
modification of the geodesic flow of the scalar manifold $(\cM,\cG)$
and a certain lift of the gradient flow of the {\em classical
effective potential} $V\eqdef M_0\sqrt{2\Phi}$ on this Riemannian
manifold. Using the invariance of oriented gradient flow orbits under
Weyl transformations of $\cG$, we found that the first order IR
approximants of cosmological orbits are insensitive to such
transformations. This allowed us to give a mathematical
description of dynamical UV and IR universality classes of classical
cosmological models. In the infrared limit, the late time dynamics of
such models is partitioned into ``phases'' which in good cases are
indexed by critical points of $V$ and by Freudenthal ends of
$\cM$. These results provide a realization of ideas familiar from the
theory of critical phenomena within the context of {\em classical}
cosmology and open up new directions for the study and classification
of such models.

Since the UV and IR approximations are controlled by different
conditions from those governing the slow roll approximation and its
slow roll - slow turn variant, the study of the IR and UV behavior of
classical cosmological models requires new ideas and tools. The UV
approximation recovers a modification of the geodesic flow of
$(\cM,\cG)$, thus making contact with a classical subject in
Riemannian geometry and dynamical systems theory. The IR approximation
recovers the gradient flow of the classical effective potential $V$,
thus making contact with another classical subject. Since the scalar
manifold $(\cM,\cG)$ is generally non-compact, one has to analyze such
flows without the compactness assumptions made in most of the
mathematics literature -- and hence the study of cosmological scaling
limits does {\em not} reduce to a simple application of known
results. For example, one cannot directly apply classical results of
Morse and gradient flow theory to the study of the IR limit even when
a natural smooth compactification of $\cM$ exists and $\Phi$ admits a
smooth and Morse extension to that compactification, because the
scalar manifold metric $\cG$ does not generally extend.

For two-field models, the target manifold $\cM$ is a connected surface
$\Sigma$. In this case, our results and the uniformization theorem of
Poincar\'e imply that the IR behavior of the model parameterized by
$(M_0,\Sigma,\cG,\Phi)$ coincides with that of the model parameterized
by $(M_0,\Sigma,G,\Phi)$ to first order in the scale expansion, where
$G$ is the uniformizing metric of $\cG$. This implies that two-field
models whose scalar manifold metric has Gaussian curvature equal to
$-1$, $0$ or $1$ are IR universal in the Wilsonian sense that they
provide distinguished representatives of the infrared universality
classes of all two-field models. This gives a powerful conceptual
reason to single out such models for detailed study. In \cite{grad},
we study the IR behavior of a very large class of two-field
cosmological models using the methods and ideas of the present paper.

One direction for further research concerns the systematic study of
higher orders of the UV and IR expansions, on which we hope to report
shortly. In this regard, it would be interesting to extract precise
asymptotic bounds which control the error terms at each order. More
ambitiously, one can look for improvements of these expansions which
produce uniformly convergent series; such improvements can be
extracted in principle using the method of multiple scales (see
\cite{BO,KC,Hinch,Kuehn}). In this context, the renormalization group
action constructed in this paper relates to the work of
\cite{Ziane,DeVille,Chiba1,Chiba2}.

It would also be interesting to study the implications of scale
expansions for cosmological perturbation theory with a view toward
developing new approximation schemes and constructing effective
descriptions which take into account the fact that the underlying
classical dynamics of cosmological models has a dynamical RG flow of
its own. In our opinion, this could address the conceptual criticism
of current approaches to effective descriptions in cosmology made in
papers such as \cite{Burgess1,Burgess2}.

Since the dynamics of multifield cosmological models admits a
constrained Hamiltonian description in the minisuperspace formalism,
the RG flow and scale approximations considered in this paper can be
reformulated in that framework, which allows one to use methods
from the theory of Hamiltonian systems. It would be interesting to see
what insight can be gained by pursuing Wilsonian ideas in that
context.

Finally, one can expect that the models with hidden symmetries
considered in \cite{Noether1,Noether2,Hesse} (see \cite{Lilia1,Lilia2}
for some phenomenological implications) play a special role in the
infrared limit, by analogy to similar situations in the theory of
critical phenomena; it would be interesting to address this
conjecture.

When the target manifold $\cM$ has dimension greater than two, the
existence of a metric $G$ with special properties in the
Weyl-equivalence class of $\cG$ is a classical problem. When $\cM$ is
non-compact and one requires $G$ to be complete and of constant scalar
curvature, this is the Yamabe problem for non-compact manifolds, which
generally has a negative answer unless one imposes further
restrictions on $(\cM,\cG)$. When $\cM$ is three-dimensional, one can
instead use Thurston uniformization, which is different in
character. It would be interesting to see if analogues of our
arguments exist that would allow one to find good representatives of
universality classes of models with more than two scalar fields when
appropriate conditions are imposed on the scalar manifold and
potential.

\appendix

\section{Topological and smooth equivalence of dynamical systems}
\label{app:tsequiv}

\noindent Recall that a (smooth) {\em autonomous dynamical
system} is a pair $(M,X)$, where $M$ is a manifold and $X$ is a
smooth vector field defined on $M$ (see \cite{Palis}).

\begin{definition}
A {\em flow curve} of the autonomous dynamical system $(M,X)$ is a
smooth integral curve $\gamma:I\rightarrow M$ of the vector field
$X$, i.e. a solution of the equation:
\be
\frac{\dd \gamma(t)}{\dd t}=X(\gamma(t))~~\forall t\in I~~,
\ee
where $I$ is a non-degenerate interval, i.e. a non-empty (open, closed
or semi-closed) interval which is not reduced to a point. The image
$\im \gamma\eqdef \gamma(I)\subset M$ of a flow curve is called a
{\em flow orbit}.
\end{definition}

\begin{definition}
The {\em flow} of the autonomous dynamical system $(M,X)$ is the flow
$\Pi:\cD\rightarrow M$ of the vector field $X$ considered on its
maximal domain of definition $\cD\subset \R\times M$.
\end{definition}

\noindent Recall that $\cD$ is an open subset of $\R\times M$ which
contains the set ${\widetilde M}\eqdef \{0\}\times M$. It is
fibered over $M$ with fiber at $m\in M$ given by the interval of
definition $I_m$ of the maximal integral curve
$\gamma_m:I_m\rightarrow M$ of $X$ which satisfies $0\in I_m$ and
$\gamma(0)=m$. Notice that the interval $I_m$ is open for all $m\in
M$. By definition, we have:
\be
\Pi(q,m)=\gamma_m(q)~,~\forall (q,m)\in \cD~~.
\ee
We denote by $\pi:\cD\rightarrow M$ the projection on the second
factor.  For each $m\in M$, we have $\Pi(0,m)=m$. Moreover, for all
$t_0\in I_m$ we have $I_{\gamma(m)(t_0)}=I_m-t_0$ and for all $t\in
I_{\gamma_m(t_0)}$ we have:
\be
\Pi(t,\Pi(t_0,m))=\Pi(t+t_0,m)~~.
\ee

\begin{definition}
Let $(M_1,X_1)$ and $(M_2,X_2)$ be autonomous dynamical systems
with flows $\Pi_k:\cD_k\rightarrow M_k$ and domain projections
$\pi_k:\cD_k\rightarrow M_k$ ($k=1,2$). Denote by
$I^{(k)}_{m_k}\eqdef \pi_k^{-1}(m_k)$ the fiber of $\cD_k$ at $m_k\in
M_k$. Let $h:M_1\rightarrow M_2$ be a homeomorphism and
$f:\cD_1\rightarrow \cD_2$ be an unbased isomorphism of topological
fiber bundles above $h$, thus:
\be
f(t,m)=(f_m(t),h(m))~~\forall m\in M_1~~\forall t\in I^{(1)}_m~~,
\ee
where $f_m:I_m^{(1)}\rightarrow I_{h(m)}^{(2)}$ is a homeomorphism for
all $m\in M_1$.  The pair $(f,h)$ is called:
\begin{enumerate}[1.]
\itemsep 0.0em
\item {\em topological equivalence}, if the following conditions are
satisfied:
\begin{itemize}
\item We have:
\ben
\label{PiEq}  
\Pi_2\circ f=h\circ \Pi_1~~,
\een
i.e. :
\be
h(\gamma_{m}^{(1)}(t))=\gamma^{(2)}_{h(m)}(f_m(t))~~\forall m\in M~~\forall t\in I_m^{(1)}~~,
\ee
where $\gamma^{(k)}_{m_k}:I_{m_k}^{(k)}\rightarrow M_k$ is the flow
curve of $\Pi_k$ which satisfies $\gamma_{m_k}^{(k)}(0)=m_k$.
\item $f_m:I_m^{(1)}\rightarrow I_{h(m)}^{(2)}$ is a strictly
increasing function for all $m\in M$.
\end{itemize}
\item {\em smooth topological equivalence}, if it is a topological
equivalence and the maps $f$ and $h$ are smooth diffeomorphisms.
\item {\em topological conjugation} if it is a topological equivalence
  and for all $m\in M$ we have $I_m^{(1)}=I_{h(m)}^{(2)}:=I_m$ and $f_m=\id_{I_m}$. 
\item {\em smooth topological conjugation} if it is a topological
conjugation and the maps $f$ and $h$ are smooth diffeomorphisms.
\end{enumerate}
The two dynamical systems and their flows are called {\em (smoothly)
topologically conjugate/equivalent} if there exists a (smooth)
topological conjugation/equivalence $(f,h)$ from $\Pi_1$ to $\Pi_2$.
\end{definition}

\noindent (Smooth) topological equivalence and conjugation are
equivalence relations. If $\Pi_1$ and $\Pi_2$ are topologically
equivalent through a pair $(f,h)$, then $h$ induces a bijection
between the sets of orbits of $\Pi_1$ and $\Pi_2$ which preserves
their orientation.

\section{Invariance of gradient flow orbits under conformal transformations}
\label{app:grad}

\begin{prop}
\label{prop:weyl}
Let $\Omega:\cM\rightarrow \R_{>0}$ be an everywhere-positive smooth
function defined on $\cM$, $\eta:I\rightarrow \cM$ be a gradient
flow curve of the scalar triple $(\cM,\cG,V)$ and $t_0\in I$ be
arbitrary. Consider the increasing reparameterization
$\tau:I\rightarrow J$ defined through:
\ben
\label{muWeyl}
\tau(t):=\tau_{\eta}(t)\eqdef \int_{t_0}^t \dd t' \Omega(\eta(t'))+C~~,
\een
where $C$ is an arbitrary constant. Then the
reparameterized curve $\eta_r\eqdef \eta\circ
\tau^{-1}:J\rightarrow \cM$ satisfies the gradient flow equation of the
scalar triple $(\cM,\Omega\cG,V)$. In particular, the orbits of the
gradient flows of $(\cM,\cG,V)$ and $(\cM,\Omega\cG,V)$ coincide.
\end{prop}

\begin{proof}
\noindent Consider the Weyl transformation:
\be
\cG\rightarrow \cG' \eqdef \Omega \cG~~,
\ee
and notice that $\grad_{\cG'}V=\frac{1}{\Omega} \grad \, V$. The
reparameterized curve $\eta_r$ satisfies:
\ben
\label{gflow}
\frac{\dd\tau}{\dd t} \frac{\dd\eta_r}{\dd \tau} +(\grad_\cG \, V)\circ\eta_r=0~~.
\een
Since $\grad_{\Omega \cG}V=\frac{1}{\Omega}\grad_\cG V$, this is
equivalent to the gradient flow equation of the scalar triple
$(\cM,\Omega \cG,V)$ iff $\frac{\dd \tau(t)}{\dd t}=\Omega(\eta(t))$,
which gives \eqref{muWeyl}.
\end{proof}

\begin{cor}
\label{cor:Weyl}
With the notations of the previous proposition, the gradient flows of
the scalar triples $(\cM,\cG,V)$ and $(\cM,\Omega\cG,V)$ are smoothly
topologically equivalent.
\end{cor}

\begin{proof}
Let $\cD$ and $\cD_r$ be the maximal domains of definition of the
gradient flows of $(\cM,\cG,V)$ and $(\cM,\Omega\cG,V)$. Consider the
map $f:\cD\rightarrow \cD_r$ defined through:
\be
f(t,m)\eqdef (\tau_{\eta_m}(t),m)~,~~\forall (t,m)\in \cD~~,
\ee
where $\eta_m$ is the gradient flow curve of $(\cM,\cG,V)$ which
satisfies $\eta_m(0)=m$. Then $f$ is a diffeomorphism and
$(f,\id_\cM)$ is a smooth topological equivalence from the gradient
flow of $(\cM,\cG,V)$ to that of $(\cM,\Omega\cG,V)$.
\end{proof}

\begin{definition}
A {\em conformal equivalence} from the scalar triple
$(\cM_1,\cG_1,V_1)$ to the scalar triple $(\cM_2,\cG_2,V_2)$ is a
smooth conformal diffeomorphism $\Psi:(\cM_1,\cG_1)\rightarrow (\cM_2,\cG_2)$
which satisfies the condition $V_2\circ \Psi=V_1$. The scalar triples
$(\cM_1,\cG_1,V_1)$ and $(\cM_2,\cG_2,V_2)$ are called {\em
conformally equivalent} if there exists a conformal equivalence from
$(\cM_1,\cG_1,V_1)$ to $(\cM_2,\cG_2,V_2)$.
\end{definition}

\noindent It is clear that conformal equivalence is an equivalence
relation on the collection of all scalar triples.

\begin{prop}
\label{prop:conformal}
Suppose that the scalar triples $(\cM_1,\cG_1,V_1)$ and
$(\cM_2,\cG_2,V_2)$ are conformally equivalent. Then the gradient
flows of $(\cM_1,\cG_1,V_1)$ and $(\cM_2,\cG_2,V_2)$ are smoothly
topologically equivalent.
\end{prop}

\begin{proof}
Let $\Psi:(\cM_1,\cG_1,V_1)\rightarrow (\cM_2,\cG_2,V_2)$ be a
conformal diffeomorphism and let $\Pi_i:\cD_i\rightarrow \cM_i$ be the
gradient flows of $(\cM_i,\cG_i,V_i)$ for $i=1,2$.  By the definition
of conformal diffeomorphisms, there exists $\Omega\in
\cC^\infty(\cM_1,\R_{>0})$ such that $\cG'\eqdef
\Psi^\ast(\cG_2)=\Omega \cG_1$. Thus $\Psi$ is an isometry from
$(\cM_1,\cG')$ to $(\cM_2,\cG_2)$ which satisfies
$\Psi^\ast(V_2)=V_1$, i.e. an isomorphism of scalar triples from
$(\cM_1,\cG',V_1)$ to $(\cM_2,\cG_2,V_2)$. Let $\Pi':\cD'\rightarrow
\cM$ be the gradient flow of $(\cM_1,\cG',V_1)$. It is easy to see
that a curve $\eta:I\rightarrow \cM_1$ is a maximal gradient flow
curve of $(\cM_1,\cG',V_1)$ which satisfies $\eta_1(0)=m$ (with
$m\in \cM_1$) iff $\eta'\eqdef \Psi\circ \eta$ is a maximal
gradient flow curve of $(\cM_2,\cG_2,V_2)$ which satisfies
$\eta'(0)=\Psi(m)$. Thus we have $\cD'=\cD_2$ and
$(\id_{\cD_2},\Psi)$ is a smooth topological equivalence from $\Pi'$
to $\Pi_2$. Since $\cG'=\Omega \cG_1$, Corollary \ref{cor:Weyl} shows
that $\Pi_1$ is smoothly topologically equivalent with $\Pi'$. Since
smooth topological equivalence of flows is an equivalence relation, we
conclude.
\end{proof}

\begin{definition}
A {\em conformal automorphism} of the scalar triple $(\cM,\cG,V)$ is a
conformal equivalence from this triple to itself.   
\end{definition}

\noindent Conformal automorphisms of $(\cM,\cG,V)$ form the stabilizer
of $V$ under the group of conformal transformations of $(\cM,\cG)$.

\begin{cor}
\label{cor:conformal} The smooth topological equivalence class of the
gradient flow of $(\cM,\cG,V)$ is invariant under conformal
automorphisms of $(\cM,\cG,V)$.
\end{cor}

\begin{proof}
Follows immediately from Proposition \ref{prop:conformal}
\end{proof}

\end{document}